%% file: main.tex
\documentclass[conference]{IEEEtran}
\IEEEoverridecommandlockouts
\usepackage{float}
\usepackage{cite}
\usepackage{amsmath,amssymb,amsfonts}
\usepackage{amsthm}
\usepackage{algorithm}
\usepackage{algorithmic}
\usepackage{graphicx}
\usepackage{caption}
\usepackage{textcomp}
\usepackage{xcolor}
\usepackage{bm}
\usepackage{booktabs}
\usepackage{enumitem}
\usepackage[bookmarks]{hyperref}
\usepackage{xurl}

\def\BibTeX{{\rm B\kern-.05em{\sc i\kern-.025em b}\kern-.08em
    T\kern-.1667em\lower.7ex\hbox{E}\kern-.125emX}}

\newtheorem{definition}{Definition}
\newtheorem{proposition}{Proposition}
\newtheorem{remark}{Remark}

\newcommand{\R}{\mathbb{R}}
\DeclareMathOperator{\SD}{SD}

\input{paper_style.tex}

\begin{document}

\title{From the Superderivative to a Relative-Gradient Edge Operator}

\author{
    \IEEEauthorblockN{Alejandro Garnung Men\'{e}ndez, M.Sc., PhD Candidate\textsuperscript{1}}
    \IEEEauthorblockA{
        \textsuperscript{1}\textit{University of Oviedo, Gij\'{o}n Polytechnic School of Engineering, Asturias, Espa\~{n}a} \\
        \texttt{uo269564@uniovi.es}, \texttt{garnungalejandro@gmail.com}
    }
}

\maketitle

\begin{abstract}
First-derivative edge detectors such as Sobel, Prewitt, Roberts and Scharr measure absolute grey-level differences. We revisit the log-gradient magnitude, a relative-gradient edge operator, and motivate it by analogy with the superderivative, a multiplicative derivative in which the independent variable is perturbed multiplicatively rather than additively. In one and two dimensions the logarithm of the superderivative is a radial logarithmic derivative; its local, origin-free, isotropic analogue is the magnitude of the logarithmic gradient. The resulting operator measures relative rather than absolute change, giving an edge-strength map that is invariant under global intensity rescaling. We present the theory, show how the operator can be implemented in two steps---a classical gradient kernel plus a pointwise division---and compare it with standard first-derivative filters on synthetic images with ground-truth edges and on the official BSDS500 benchmark. We also study a parametric family that interpolates between the absolute gradient and the log-gradient; empirically, the best members are small intermediate exponents, which edge the classical kernels by a small but statistically detectable margin in the per-image metric and, with the same Gaussian smoothing as the single-$\sigma$ Canny, match it in ODS and lead it in OIS and AP, while the tuned Canny remains ahead in ODS. The single-threshold margin over the kernels does not replicate on the held-out test split, and the fully relative endpoint is never the best member, its ranking depending on the edge-map normalisation. The contribution is the construction and the invariance it buys, not a new state of the art.
\end{abstract}


\section{Introduction}
\label{sec:intro}
\IEEEPARstart{E}{dge} detection is one of the oldest and most studied tasks
in image processing.  Classical first-derivative operators --- Sobel,
Prewitt, Roberts, Scharr, and many others --- estimate the spatial gradient
of the image intensity and declare edges where the gradient magnitude is
large.  All of them share the same assumption: a grey-level jump of a given
size is equally important anywhere in the image.

That assumption is questionable.  A jump of ten grey levels on a background
of intensity five is visually much stronger than the same jump on a
background of intensity two hundred.  The human visual system responds
approximately logarithmically to light\cite{fechner}, so the relevant quantity
is the \emph{relative} jump, not the absolute one.  Several authors have
incorporated this idea into tone mapping, enhancement and denoising, usually
by working on the logarithm of the image.  Homomorphic filtering separates
illumination from reflectance by taking logarithms before filtering and
exponentiating afterwards --- a multiplicative image model $I=R\cdot L$ is
turned into an additive one by $\log I=\log R+\log L$, filtered, and mapped
back\cite{homomorphic}; Retinex theory builds lightness
from ratios of neighbouring intensities, so that the recovered quantity is a
relative rather than an absolute one\cite{land1971,land1977}, and its later
variational and centre/surround formulations keep that
property\cite{jobson1997,kimmel2003}; tone-mapping operators compress
dynamic range on a logarithmic or otherwise relative
scale\cite{reinhard,drago2003}.  All of these act on the image before a derivative is
taken.  The operator proposed here puts the same relative-change principle
into the derivative itself.

The \emph{superderivative} is the derivative obtained by perturbing the
independent variable multiplicatively rather than
additively~\cite{zotkin}; in multiplicative (non-Newtonian) calculus the same
construction is the bigeometric derivative~\cite{grossman}.  Whereas the
ordinary derivative adds a small increment $h$ to the argument, $x\mapsto x+h$,
the superderivative multiplies the argument by a factor close to one,
$x\mapsto x\cdot(1+h)$.  The resulting quantity therefore measures the
\emph{relative} growth of $f$ with respect to a proportional change of its
input, rather than the absolute growth with respect to an additive change.
The ordinary derivative admits other generalisations as well: fractional calculus
extends it to arbitrary real order through the Riemann--Liouville, Caputo and
Gr\"unwald--Letnikov operators~\cite{samko1993,kilbas2006,miller1993}, whereas
the superderivative changes the perturbation itself. The two generalisations are
independent and should not be confused: fractional calculus keeps the additive
perturbation $x\mapsto x+h$ and makes the \emph{order} real, while the
superderivative keeps the order at one and makes the \emph{perturbation}
multiplicative; the superderivative is therefore \emph{not} a derivative of
arbitrary order, and its logarithm is an ordinary first-order derivative of
$\log f$.
We use this relative-change principle to motivate a family of first-derivative
edge operators.  The simplest member of the family is the
log-gradient magnitude
\begin{equation}
E(x,y)=\|\nabla\log I(x,y)\|=\frac{\|\nabla I(x,y)\|}{I(x,y)},
\label{eq:loggrad}
\end{equation}
which replaces the absolute gradient with the relative gradient.  More
generally, the original two-dimensional superderivative score
\begin{equation}
L(x,y)=\frac{x\,\partial_x I(x,y)+y\,\partial_y I(x,y)}{I(x,y)}
\label{eq:L}
\end{equation}
measures relative change under a radial zoom centred at the origin.

Our contributions are:
\begin{itemize}[leftmargin=1.4em]
  \item A construction motivating the log-gradient edge map as the local,
        isotropic, origin-free analogue of the superderivative, together with the
        two-dimensional radial score it comes from.  The operator is motivated
        by, not derived from, the multiplicative derivative: the origin-dependent
        weights of the latter are what its isotropic form drops.
  \item A proof of invariance under global intensity rescaling.
  \item A parametric family $E_p=\|\nabla I\|/I^p$ that interpolates
        between the absolute gradient ($p=0$) and the log-gradient ($p=1$).
  \item A discussion of why the operator is not, in its full radial form, a
        pure convolution, and how it can nevertheless be evaluated in two
        efficient steps.
  \item Experimental comparisons with Sobel, Prewitt, Roberts, Scharr and
        Canny on synthetic ground-truth images, natural images, and BSDS500.
\end{itemize}

Applications beyond edge detection are outside the present scope.

\section{Background: classical first-derivative kernels}
\label{sec:kernels}
All classical first-derivative operators estimate the same quantity, the
spatial gradient, but use different finite-difference masks.  Table~\ref{tab:kernels}
lists the most common $3\times3$ (and $2\times2$) kernels\cite{gonzales}.

\begin{table*}[!t]
\centering
\footnotesize
\caption{Common first-derivative kernels.  $G_x$ estimates the horizontal
and $G_y$ the vertical derivative.  Roberts uses the two $2\times2$ masks
shown; the remaining operators define $G_y$ as the transpose of $G_x$.}
\label{tab:kernels}
\setlength{\tabcolsep}{3pt}
\begin{tabular}{@{}lll@{}}
\toprule
Operator & $G_x$ / masks & Reference \\
\midrule
Roberts &
$G_x=\begin{bmatrix} 1 & 0 \\ 0 & -1 \end{bmatrix}$,
$G_y=\begin{bmatrix} 0 & 1 \\ -1 & 0 \end{bmatrix}$ & \cite{roberts} \\
\addlinespace
Prewitt &
$\begin{bmatrix} -1 & 0 & 1 \\ -1 & 0 & 1 \\ -1 & 0 & 1 \end{bmatrix}$ & \cite{prewitt} \\
\addlinespace
Sobel &
$\begin{bmatrix} -1 & 0 & 1 \\ -2 & 0 & 2 \\ -1 & 0 & 1 \end{bmatrix}$ & \cite{sobel} \\
\addlinespace
Frei--Chen / Isotropic &
$\begin{bmatrix} -1 & 0 & 1 \\ -\sqrt{2} & 0 & \sqrt{2} \\ -1 & 0 & 1 \end{bmatrix}$ & \cite{freichen} \\
\addlinespace
Scharr &
$\begin{bmatrix} -3 & 0 & 3 \\ -10 & 0 & 10 \\ -3 & 0 & 3 \end{bmatrix}$ & \cite{scharr} \\
\bottomrule
\end{tabular}
\end{table*}

Several other designs exist.  The \emph{compass operators} of Kirsch
\cite{kirsch} and Robinson \cite{robinson} use a set of oriented
$3\times3$ masks (typically eight directions) and take the maximum response
as the edge strength.  The Farid--Simoncelli filters use larger, separable
$5\times5$ kernels that are optimally steerable and have better rotation
invariance\cite{farid,farid1997}.

The gradient magnitude is $E=\sqrt{(G_x*I)^2+(G_y*I)^2}$, where $*$ denotes
convolution.  The Canny detector adds Gaussian smoothing, non-maximum
suppression and hysteresis thresholding to a gradient estimate; it is a
pipeline rather than a single kernel\cite{canny}.  The scale-space framework
that underlies multiscale edge detection is developed by
Lindeberg~\cite{lindeberg1994}.  All of these methods treat
a given absolute grey-level jump as equally significant regardless of the
local mean.  Indeed, consider a one-dimensional step from a baseline $I_0$ to
$I_0+\Delta I$.  The derivative is the same for any $I_0$ because
$\partial_x(I_0+\Delta I\,H(x))=\Delta I\,\delta(x)$, where $H$ is the
Heaviside step and $\delta$ the Dirac delta.  The measured edge strength is
therefore proportional to $|\Delta I|$ and independent of $I_0$.  Equivalently, at a
fixed \emph{contrast ratio} $\Delta I/I_0$ the absolute response grows with the
background $I_0$, so a fixed threshold detects equal-contrast edges unequally; this
is the regime in which the relative operator of Section~\ref{sec:operator}
differs from the classical one, and it is tested directly in
Section~\ref{sec:relative}.

\section{The superderivative}
\label{sec:superderiv}
\subsection{One-dimensional definition}
For a positive function $f$, the superderivative is defined by~\cite{zotkin,grossman}
\begin{equation}
\SD(f(x))=\lim_{h\to0}\sqrt[h]{\frac{f(x\cdot(1+h))}{f(x)}}
=\exp\!\Big(x\,\frac{f'(x)}{f(x)}\Big).
\label{eq:sd}
\end{equation}
The limit follows from the first-order expansion
$f(x\cdot(1+h))=f(x+xh)\approx f(x)+xh f'(x)$, which gives
\begin{equation}
\frac{f(x\cdot(1+h))}{f(x)}\approx 1+xh\,\frac{f'(x)}{f(x)}.
\end{equation}
Write $a=x\,\frac{f'(x)}{f(x)}$ and assume $a\neq0$ (if $a=0$ the limit is
$1=e^0$).  Then
\begin{equation}
\Bigl(1+xh\,\frac{f'(x)}{f(x)}\Bigr)^{1/h}
=(1+ah)^{1/h}
=\Bigl[(1+ah)^{1/(ah)}\Bigr]^{a}.
\end{equation}
As $h\to0$ we have $ah\to0$, and the classical limit
$\lim_{t\to0}(1+t)^{1/t}=e$ gives $(1+ah)^{1/(ah)}\to e$.  Therefore the
whole expression tends to $e^{a}=\exp\!\bigl(xf'(x)/f(x)\bigr)$.
Taking logarithms (valid because $f>0$ implies $\SD(f)>0$) gives the more
convenient identity
\begin{equation}
\log\SD(f(x))=x\,\frac{f'(x)}{f(x)}=x\,\partial_x\log f(x).
\label{eq:logsd}
\end{equation}
The last equality is the chain rule: since $\partial_u\log u=1/u$,
$\partial_x\log f(x)=\frac{1}{f(x)}\cdot f'(x)=\frac{f'(x)}{f(x)}$.
For a power law $f(x)=x^k$ one obtains $\SD(f)=e^k$: the superderivative
extracts the exponent $k$ of homogeneous growth, i.e. the number $k$ for which
$f(cx)=c^k f(x)$ for every scaling factor $c>0$.  This is the sense in which
it measures relative change; the name \emph{superderivative} is used here in the
sense of multiplicative calculus~\cite{zotkin} and is unrelated to the
superderivative of supergeometry (a derivation on a superalgebra).

\subsection{Two-dimensional extension}
\begin{definition}[Spatial superderivative score]
Let $I:\Omega\to\R_{>0}$ be a positive image on a domain
$\Omega\subset\R^2$ centred at the origin.  The \emph{superderivative score}
is
\begin{equation}
L(x,y)=\frac{x\,\partial_x I(x,y)+y\,\partial_y I(x,y)}{I(x,y)},
\label{eq:L2}
\end{equation}
and the 2-D superderivative is $\SD(I)(x,y)=\exp(L(x,y))$.
\end{definition}

In polar coordinates, Eq.~\eqref{eq:L2} becomes $L=r\,\partial_r\log I$, the
radial logarithmic derivative.  The score is invariant under a global
intensity rescaling $I\mapsto cI$ and, after recentring, covariant under a
global spatial scaling: under $I(\cdot)\mapsto I(\lambda\,\cdot)$ it is
reparametrised as $L'(x)=L(\lambda x)$, so its value at corresponding
points is unchanged.

\subsection{From radial superderivative to isotropic log-gradient}
\label{sec:synthesis}
The one-dimensional superderivative Eq.~\eqref{eq:sd}~\cite{zotkin,grossman} measures relative growth
along the ray emanating from the origin: the two evaluation points $x$ and
$x(1+h)$ lie on the same ray, and the quotient $f(x(1+h))/f(x)$ is the relative
(multiplicative) change of $f$ when the argument is scaled by $1+h$ along that
ray.  Taking the $1/h$ power and the limit converts this ratio into an
instantaneous relative growth rate.  In two dimensions the same idea
gives the \emph{radial superderivative}
\begin{multline}
\SD_r(I)(r,\theta)
= \lim_{h\to0} \sqrt[h]{\frac{I(r\cdot(1+h),\theta)}{I(r,\theta)}} \\
= \exp\!\big(r\,\partial_r\log I(r,\theta)\big).
\label{eq:sdradial}
\end{multline}
Because $r\,\partial_r = x\,\partial_x + y\,\partial_y$, returning to
Cartesian coordinates yields exactly the spatial score
\begin{equation}
\log \SD_r(I)(x,y)
= \frac{x\,\partial_x I + y\,\partial_y I}{I}
= L(x,y),
\label{eq:Lcartesian}
\end{equation}
i.e. Definition~\eqref{eq:L2} is the direct two-dimensional transcription of
the original definition~\cite{zotkin,grossman}.

For edge detection the radial score has two shortcomings: it depends on the
choice of origin through the weights $x$ and $y$, and it measures change only
in the radial direction.  We therefore seek a local, origin-free and
direction-independent measure of the same relative-change quantity.

The logarithmic derivative in one dimension is
$\partial_x\log f = f'/f$.  Its natural two-dimensional extension is the
logarithmic gradient
\begin{equation}
\nabla\log I = \frac{\nabla I}{I}.
\label{eq:loggradvec}
\end{equation}
Taking the Euclidean norm gives
\begin{equation}
\|\nabla\log I\| = \frac{\|\nabla I\|}{I},
\label{eq:normloggrad}
\end{equation}
which measures the total relative rate of change of $I$ without privileging
any direction.  For a purely radial image $I=I(r)$ we have
$\nabla I=(\partial_r I)\,\mathbf{e}_r$, and therefore
\begin{equation}
\|\nabla\log I\| = |\partial_r\log I| = \frac{|L(x,y)|}{r}.
\label{eq:radialrelation}
\end{equation}
Thus the isotropic operator equals the radial logarithmic derivative
$|\partial_r\log I| = |L(x,y)|/r$ of the image intensity --- not
the radial superderivative itself, which is $\exp(L)$.  The isotropic operator
is therefore a natural local, origin-free analogue of the original radial
superderivative Eq.~\eqref{eq:sd}, though not a derived counterpart: the
superderivative is an exponential (multiplicative) object, so its exact isotropic
multiplicative analogue would be $\exp(\|\nabla\log I\|)$, whereas the isotropic
operator is the additive logarithmic gradient that underlies it.  We keep the
isotropic operator, written $E_{\mathrm{SD}}$ in Eq.~\eqref{eq:Esd}, because it is
the quantity that is directly comparable to the classical gradient magnitude, and we
use it as an analogue motivated by the same relative-change principle rather than as
a unique extension of the superderivative.  In one dimension the distinction is explicit: dropping the weight
$x$ from $\log\SD(f)=x\,f'/f$ leaves the logarithmic derivative
$\partial_x\log f=f'/f$, which is the logarithm of the \emph{geometric}
derivative $\exp(f'/f)$, not of the bigeometric one.  The isotropic operator is
thus a per-direction geometric log-derivative, and the superderivative contributes
the relative-change principle and the radial score rather than the exact operator.

\subsection{Why it is not a convolution}
A standard first-derivative filter such as Sobel is a convolution: the same
kernel is applied at every pixel.  The radial score $L$ cannot be expressed
as a convolution because the weights $x$ and $y$ depend on the global
coordinates.  However, as we show below, a local and isotropic operator
with the same relative-change semantics can be obtained by combining a
classical derivative kernel with a pointwise division.

\section{The superderivative edge operator}
\label{sec:operator}
\subsection{The isotropic log-gradient edge map}
Section~\ref{sec:synthesis} motivated the magnitude of the logarithmic
gradient as the local, origin-free, isotropic analogue of the
radial superderivative.  We
define the \emph{relative-gradient edge-strength map} $E_{\mathrm{SD}}$ as
\begin{equation}
E_{\mathrm{SD}}(x,y)=\|\nabla\log I(x,y)\|
=\frac{\sqrt{(\partial_x I)^2+(\partial_y I)^2}}{I}.
\label{eq:Esd}
\end{equation}
It is the classical gradient magnitude divided pointwise by the local
intensity.  It should not be confused with the canonical radial score $L$
of Eq.~\eqref{eq:L2}: $E_{\mathrm{SD}}$ is an isotropic operator inspired by
the same relative-change principle, while $L$ is the direct two-dimensional
transcription of the original superderivative.

\paragraph*{Relation to existing work}
The quantity $\|\nabla\log I\|$ is not new in itself, and it is worth saying
so plainly.  It is the magnitude of the logarithmic (or relative) gradient,
the differential form of the Weber contrast $\Delta I/I$ (the relative luminance
change, named after E.~H.~Weber, that the eye detects against a background
$I$)\cite{fechner}, and it
is what a gradient operator computes after the logarithmic transform used in
homomorphic filtering\cite{homomorphic}, in logarithmic image
processing\cite{jourlin}, and in tone mapping\cite{reinhard}.  Multiplicative
(non-Newtonian) calculus, in which the geometric derivative of a positive
function is $\exp(f'/f)$ and the bigeometric derivative is
$\exp(xf'/f)$~\cite{grossman,bashirov2008}, has been developed as a self-contained
theory and applied to image analysis by Florack and van Assen\cite{florack},
who use the multiplicative derivative in diffusion-tensor imaging and strain
analysis; they do not, however, propose the log-gradient as an edge measure.
Two further families of illumination- and contrast-robust detectors are related.
Phase congruency detects features where the Fourier components are maximally in
phase and is invariant to image contrast\cite{kovesi}; the ratio-of-averages
detector for synthetic-aperture radar images is designed for multiplicative
speckle noise\cite{touzi}.  Neither is a first-derivative kernel of the image
intensity: phase congruency normalises a frequency-domain local energy, and the
ratio-of-averages detector uses the ratio of local means over two regions rather
than a pointwise logarithmic derivative.  A systematic comparison with them is left
to future work.
What the present paper adds is not the operator but the route to it:
$E_{\mathrm{SD}}$ is motivated as the local, origin-free analogue of the
radial superderivative~\cite{zotkin,grossman} (Eq.~\eqref{eq:L2}), which fixes
the two defects of that score --- its dependence on the choice of origin and its
blindness to tangential change --- and the parametric family of
Section~\ref{sec:parametric}, which shows that the fully relative operator
$p=1$ is not the best member of its own family.

\begin{proposition}[Intensity invariance]
$E_{\mathrm{SD}}$ is invariant under a global intensity rescaling
$I\mapsto cI$ for any $c>0$.
\end{proposition}
\begin{proof}
$\nabla\log(cI)=\nabla(\log c+\log I)=\nabla\log I$.
\end{proof}

\begin{proposition}[Positivity]
If $I>0$ everywhere, then $E_{\mathrm{SD}}\ge0$ everywhere.
\end{proposition}

\subsection{Relation to the radial score}
The radial score and the log-gradient are two facets of the same idea.  In
polar coordinates,
\begin{align}
L &= r\,\partial_r\log I, \\
\|\nabla\log I\| &= \sqrt{(\partial_r\log I)^2+(\partial_\theta\log I/r)^2}.
\end{align}
Thus $|L|/r=|\partial_r\log I|$ is the radial component of the log-gradient,
while $E_{\mathrm{SD}}$ is its full magnitude and therefore also captures the
angular component.  For images with
strong radial symmetry, $|L|$ or $|L|/r$ may be preferred; for general edge
detection we use $E_{\mathrm{SD}}$.  Figure~\ref{fig:radial} compares the
canonical radial score $L$, its normalised form $|L|/r$, and the isotropic
operator $E_{\mathrm{SD}}$ on two natural images, one radially structured about
its centre and one not.  On the retinal fundus the radial forms are adapted to
the symmetry and follow the vessel arcade; on the astronaut the isotropic
operator is more generally applicable.  At an intensity step the canonical
radial score $L$ gives a single signed response because it is a directed
derivative, whereas the magnitude forms $|L|/r$ and $E_{\mathrm{SD}}$ take
absolute values and, through the $1/I$ (or $1/r$) weight, emphasise the dark side
of the step; the result remains a single, generally asymmetric peak rather than a
double-lobed or second-derivative-like profile.

\begin{figure*}[!t]
\centering
\includegraphics[width=\textwidth]{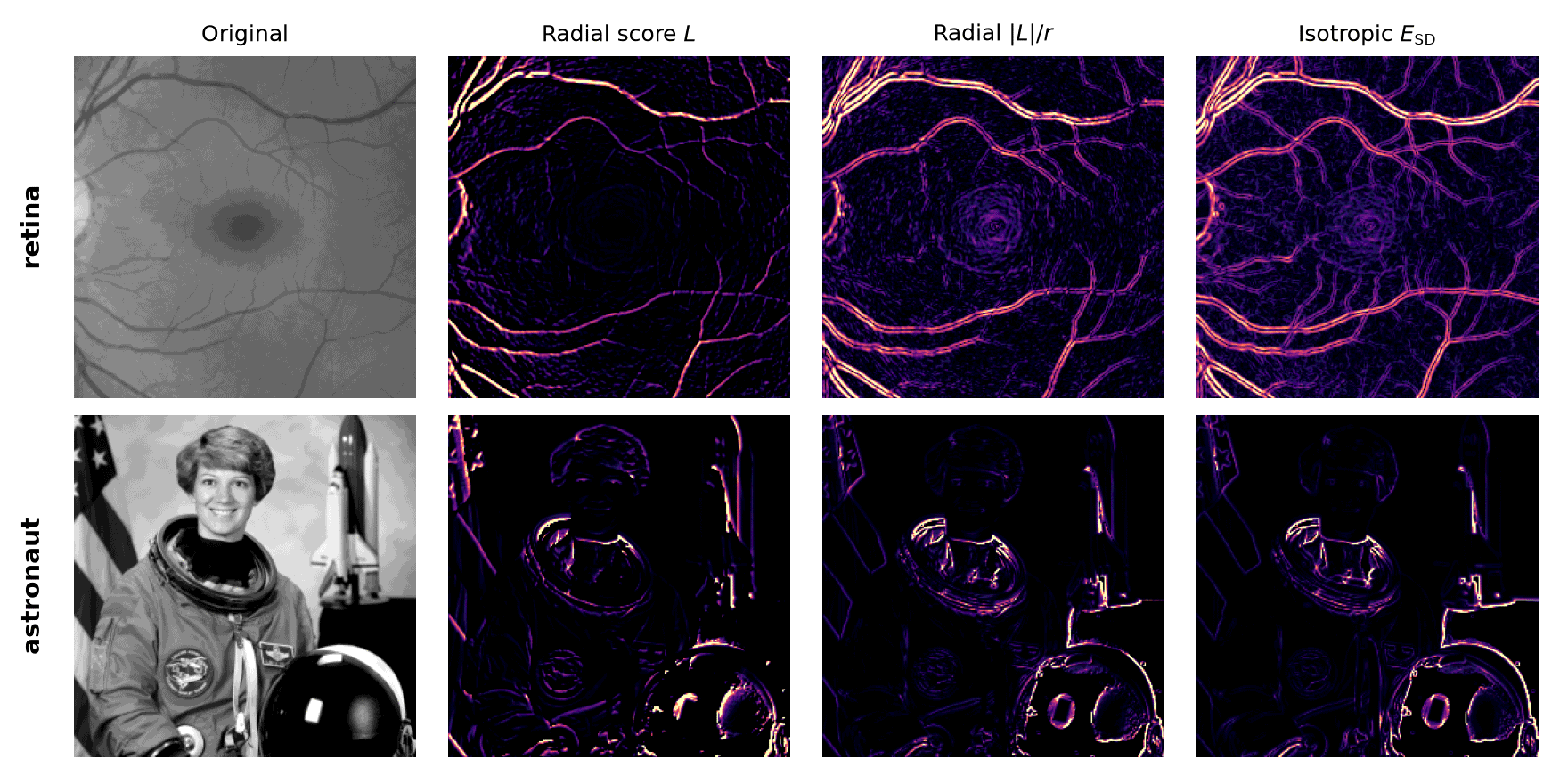}
\caption{Comparison of the canonical radial score $L=(xI_x+yI_y)/I$, its
normalised version $|L|/r$, and the isotropic superderivative
$E_{\mathrm{SD}}=\|\nabla I\|/I$ (Eq.~\eqref{eq:Esd}).  The retinal fundus is
radially structured about the image centre (radial illumination falloff, macula at
the centre, vessel arcade sweeping around it), which is the regime the radial forms
assume; the astronaut has no such symmetry, and the isotropic form is the
recommended operator for general edge detection.}
\label{fig:radial}
\end{figure*}

\subsection{Implementation}
The operator can be evaluated in two steps.  Algorithm~\ref{alg:sd} gives
a concise description; any classical derivative kernel can be used in
line~\ref{ln:deriv}.

\begin{algorithm}[H]
\caption{Isotropic relative-gradient edge-strength map (implements
Eq.~\eqref{eq:Esd}, not the radial score $L$ of Eq.~\eqref{eq:L2}).}
\label{alg:sd}
\begin{algorithmic}[1]
\REQUIRE Greyscale image $I>0$, derivative kernel $K$
\ENSURE Edge-strength map $E_{\mathrm{SD}}$
\STATE Compute $\partial_x I = K * I$ and $\partial_y I = K^T * I$
        \label{ln:deriv}
\STATE $E_{\mathrm{SD}} \leftarrow \sqrt{(\partial_x I)^2+(\partial_y I)^2}\,/\,I$
\RETURN $E_{\mathrm{SD}}$
\end{algorithmic}
\end{algorithm}

Step~\ref{ln:deriv} is a convolution; step~2 is pointwise.  The whole
operator therefore has the same computational complexity as a classical
gradient magnitude plus one division per pixel.  The two-step procedure estimates
Eq.~\eqref{eq:Esd} with the same finite-difference kernel as a classical
gradient magnitude: the kernel provides an estimate of $\nabla I$, and the
pointwise division by $I$ yields $\|\nabla I\|/I$.  The only approximation
inherited from classical methods is the choice of finite-difference kernel.

The two factors of Eq.~\eqref{eq:Esd} can also be discretised in the opposite
order.  Since $\|\nabla I\|/I=\|\nabla\log I\|$ in the continuum, one may instead
apply the kernel to the logarithmic image and take $\|K*\log I\|$.  The two
estimates agree away from a step but differ at one: the first weights the
intensity difference by the value at a single pixel, so at a step from $I_1$ to
$I_2$ it gives $|I_2-I_1|/I_1$ on one side and $|I_2-I_1|/I_2$ on the other,
whereas the second gives the symmetric logarithmic difference
$|\log(I_2/I_1)|$ on both sides.  The two-step form is the default because it
reuses the classical kernels directly; Section~\ref{sec:discretisation} compares
the two.

A practical caveat is that real images may contain zero-valued pixels.  A
small offset $\varepsilon$ can be added to the denominator, giving
$E_{\varepsilon}=\|\nabla I\|/(I+\varepsilon)$.  This breaks exact intensity
invariance in the darkest regions but avoids numerical instability.

\section{A parametric family}
\label{sec:parametric}
The absolute gradient and the log-gradient are two extreme points of a single
parametric family:
\begin{equation}
E_p(x,y)=\frac{\|\nabla I(x,y)\|}{I(x,y)^p},
\qquad p\in[0,\infty).
\label{eq:parametric}
\end{equation}
For $p=0$ we recover the classical absolute gradient magnitude.  For $p=1$ we
recover the superderivative/log-gradient magnitude.  Values $p\in(0,1)$ give
an intermediate behaviour: the response is partially normalised by the local
intensity without becoming as aggressive as the full log-gradient in dark
regions.

\begin{proposition}[Intensity invariance of the family]
$E_p$ is invariant under $I\mapsto cI$ if and only if $p=1$.
\end{proposition}
\begin{proof}
$\|\nabla(cI)\|=c\|\nabla I\|$ and $(cI)^p=c^p I^p$.  Therefore
$E_p(cI)=c^{1-p}\|\nabla I\|/I^p=c^{1-p}E_p(I)$, which equals $E_p(I)$ if and
only if $p=1$.
\end{proof}

\begin{remark}[Gamma transforms]
\label{rem:gamma}
Since $\nabla\log(I^{\gamma})=\gamma\nabla\log I$, a gamma transform multiplies
$E_{\mathrm{SD}}$ by the constant $\gamma$ and leaves the normalised edge map
unchanged, while for $p\neq1$ it gives the position-dependent factor
$\gamma\,I^{(\gamma-1)(1-p)}$.  The $p=1$ operator is thus covariant under
gamma, not merely scale-invariant --- one more property pulling against the
experiments, which prefer a small interior exponent ($p=0.25$).
\end{remark}

Thus the full log-gradient ($p=1$) is the unique intensity-invariant member of
the family; members with $p<1$ trade some invariance for numerical stability.

The above observation is important in practice: exact intensity invariance
requires $p=1$, but the BSDS500 experiments below show that the best empirical
performance is obtained with a small interior exponent ($p=0.25$).
The slight loss of invariance is compensated by better stability in dark regions.

\subsection{Regularisation and robust alternatives}
\label{sec:regularised}
The raw log-gradient can become numerically unstable or visually faint when
$I$ is small or when the relative change is very small.  Several regularised
variants keep the same relative-change semantics while improving robustness:

\begin{itemize}[leftmargin=1.4em]
  \item \textbf{Additive regularisation:}
    $E_{\alpha,p}=\|\nabla I\|/(I^p+\alpha)$ with $\alpha>0$.  In bright
    regions this behaves like $E_p$; in dark regions it behaves like an
    absolute gradient scaled by $1/\alpha$, which avoids division-by-zero
    and prevents the map from collapsing to black.
  \item \textbf{Local-mean denominator:} replace $I^p$ by
    $(\mu_\rho*I)^p$, where the locally averaged intensity $\mu_\rho*I$
    is computed with a box/Gaussian filter of radius
    $\rho$ and then raised to the same exponent.  This makes the denominator robust to shot noise while retaining
    the relative-change interpretation at scales larger than $\rho$.
\end{itemize}

Figure~\ref{fig:formulations} compares these variants on natural images.

\begin{figure*}[!t]
\centering
\includegraphics[width=\textwidth]{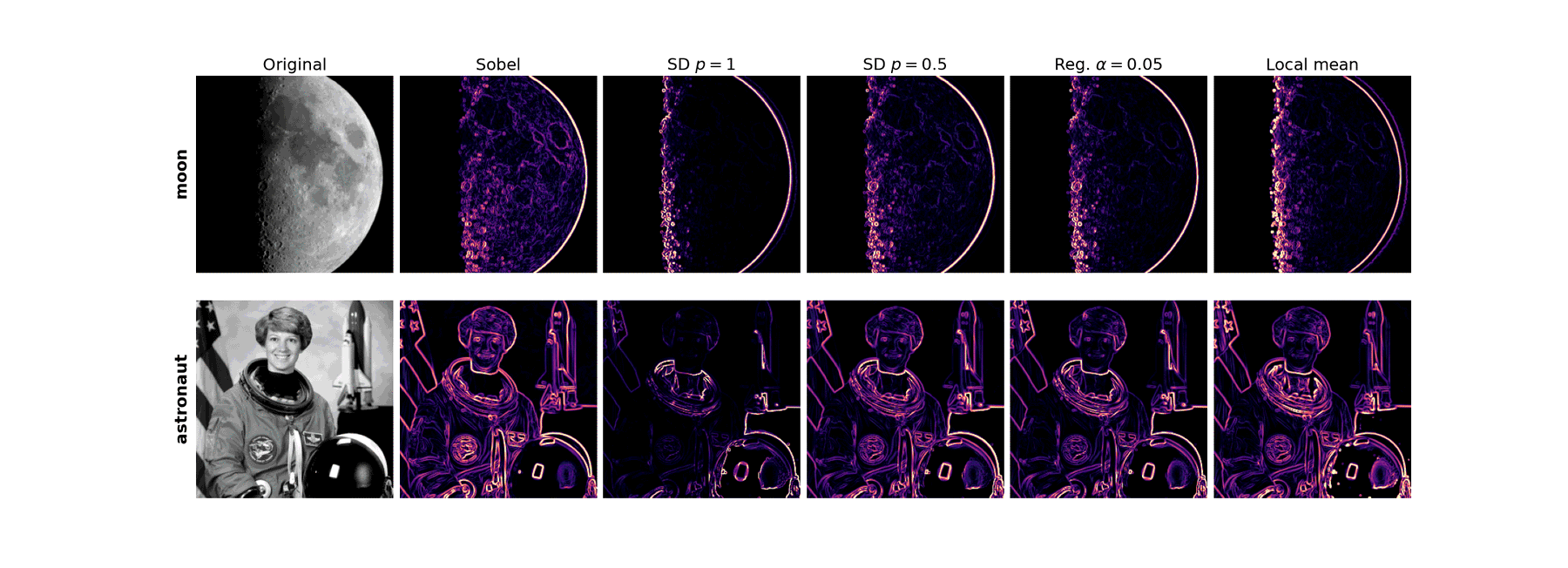}
\caption{Alternative superderivative formulations, one row per image and one
column per formulation: original, Sobel, the log-gradient $p=1$, the
parametric $p=0.5$, the additively regularised $\alpha=0.05$ and the
local-mean denominator.  The lunar image is included because its dark sky is
where an unregularised division by $I$ is worst behaved, which is what the
regularised and local-mean variants are for.  Maps are scaled to their 99th
percentile, so a single saturated pixel in the dark sky cannot blacken the
rest of the panel.}
\label{fig:formulations}
\end{figure*}

\section{Experiments}
\label{sec:experiments}
\subsection{Setup}
We generated five synthetic $256\times256$ images with known ground-truth
edges:
\begin{itemize}[leftmargin=1.4em]
  \item \textbf{Circles:} concentric annuli with five boundaries.
  \item \textbf{Square:} a bright square on a dark background.
  \item \textbf{Checkerboard:} an $8\times8$ checkerboard.
  \item \textbf{Triangle:} a bright triangle on a dark background.
  \item \textbf{Ramp:} a linear intensity ramp with an additive step.
\end{itemize}
The annulus boundaries, the triangle boundary, and the square boundary are
computed analytically from the generating mask, so the ground truth does not depend
on any detector; the checkerboard ground truth is the grid of internal boundaries of
the generating mask, and the ramp ground truth is the column of the additive step.
For each image we computed the edge-strength maps of the proposed operator,
the radial score $|L|/r$, and the classical operators Sobel, Scharr, Prewitt,
Roberts and Canny.  The proposed and radial operators use the Sobel kernel
throughout, while each classical operator uses its own kernel, so that the parametric
comparison is an ablation against the same kernel for the proposed family; the input is
clipped to $[10^{-3},1]$, each map is scaled to $[0,1]$ by its minimum and
maximum, thresholded at 23 levels linearly spaced in $[0.02,0.90]$, and thinned to a
one-pixel skeleton (scikit-image \texttt{skeletonize}) before matching.  The best F1 over the threshold sweep was
computed with a one-pixel localisation tolerance, as in the
BSDS benchmark (Section~\ref{sec:bsds}).  Here $F_1=2PR/(P+R)$ is the harmonic
mean of precision $P$ and recall $R$ at a given threshold, and the ``best'' is
taken over the threshold sweep; a one-pixel localisation tolerance credits a
predicted edge that falls within one pixel of the ground truth, instead of
requiring exact coincidence.  A strict pixel-exact rule is not used
to rank detectors, because it charges a one-pixel offset as both a false positive
and a false negative.

\subsection{Quantitative results}
Table~\ref{tab:results} reports F1 scores.  On these piecewise-constant
images most first-derivative methods, including the proposed one, recover the
edges almost perfectly.  The main exceptions are Roberts on the triangle, whose
$2\times2$ response misses part of the boundary ($0.871$), and the radial
score, which fails on the ramp, whose gradual slope is tangential to the radial
direction over most of the image, and loses on the checkerboard, whose axis-aligned
internal boundaries are mostly perpendicular to the radial direction away from the
centre.  A strict pixel-exact rule
would charge the one-pixel offset of detectors such as Roberts, and of Canny on
the ramp, as both a false positive and a false negative, an artefact of the
metric rather than a property of the operators; we therefore report only the
tolerance-based rule throughout.  Because these piecewise-constant images cannot
separate the operators, Section~\ref{sec:relative} isolates the regime where the
relative and absolute weights must differ.

\begin{table*}[!t]
\centering
\footnotesize
\caption{Best F1 on synthetic ground-truth edges under a threshold sweep with
one-pixel localisation tolerance (Section~\ref{sec:bsds}).
Highest score per image is in \textbf{bold}.  Most first-derivative methods are
within about $0.02$ of each other; the exceptions are Roberts on the triangle and
the radial score on the ramp and checkerboard, where their respective definitions
are a genuine handicap.}
\label{tab:results}
\setlength{\tabcolsep}{2pt}
\begin{tabular}{@{}lccccccc@{}}
\toprule
Image & $E_{0.25}$ (Sobel) & $|L|/r$ & Sobel & Scharr & Prewitt & Roberts & Canny \\
\midrule
Circles      & 0.986 & 0.985 & 0.986 & \textbf{0.988} & 0.984 & 0.979 & 0.980 \\
Square       & \textbf{1.000} & \textbf{1.000} & 0.998 & 0.998 & 0.998 & 0.998 & 0.996 \\
Checkerboard & \textbf{1.000} & 0.942 & 0.998 & 0.998 & 0.998 & \textbf{1.000} & 0.998 \\
Triangle     & \textbf{1.000} & \textbf{1.000} & \textbf{1.000} & \textbf{1.000} & \textbf{1.000} & 0.871 & 0.998 \\
Ramp         & 0.996 & 0.188 & 0.996 & 0.996 & 0.996 & \textbf{1.000} & 0.996 \\
\bottomrule
\end{tabular}%

\end{table*}

\subsection{Visual comparison}
Figure~\ref{fig:edges} shows the edge-strength maps.  The superderivative
operator gives thin, clean responses at intensity discontinuities, similar
to Sobel and Scharr.  Because it normalises by the local intensity, its
response is more uniform across objects of different brightness.

\begin{figure*}[!t]
\centering
\includegraphics[width=\textwidth]{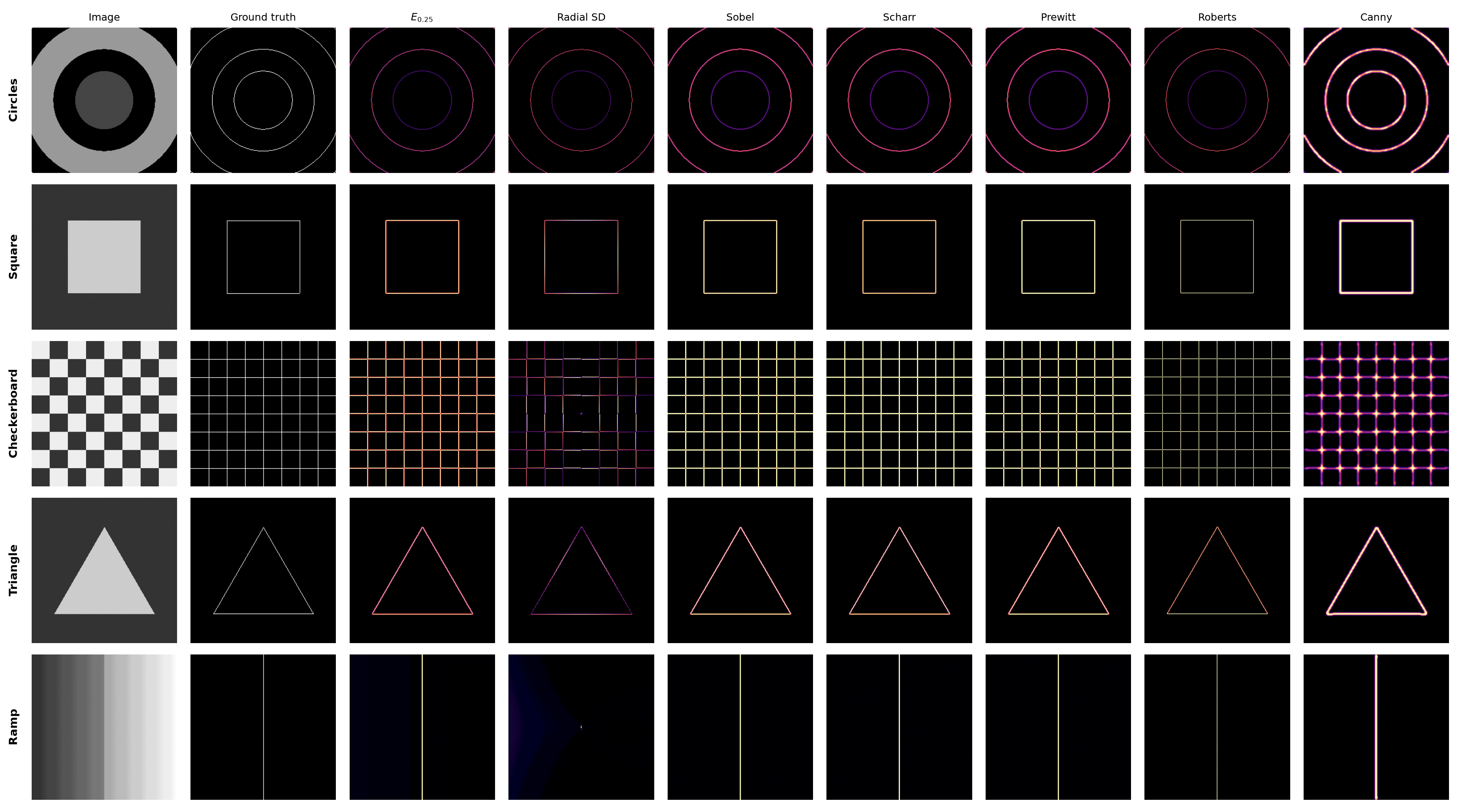}
\caption{Edge-strength maps for synthetic images.  Rows correspond to
concentric annuli, square, checkerboard, triangle and intensity ramp; columns
show the original image, ground truth, and the outputs of the superderivative
operator ($p=0.25$), the radial score, Sobel, Scharr, Prewitt, Roberts and Canny.
The radial score on the ramp collapses to a single spot because the gradual
slope contributes only a weak, coordinate-weighted response, while the additive
step produces a strong localised peak that dominates the normalised colour
scale.}
\label{fig:edges}
\end{figure*}

\subsection{Intensity-scale invariance}
To verify the main theoretical property experimentally, we took a synthetic
piecewise-constant image and multiplied it by $c\in\{0.3,1.0,3.0\}$.  The
Sobel edge-strength map scales linearly with $c$ ($\|\nabla(cI)\|=c\|\nabla
I\|$), whereas the superderivative map remains unchanged because the factor
$c$ cancels in $\|\nabla(cI)\|/(cI)$.  Figure~\ref{fig:invariance} shows the
maps on a common absolute scale: the Sobel maps become brighter as $c$
increases, while the superderivative maps are visually identical.

\begin{figure}[!t]
\centering
\includegraphics[width=\columnwidth]{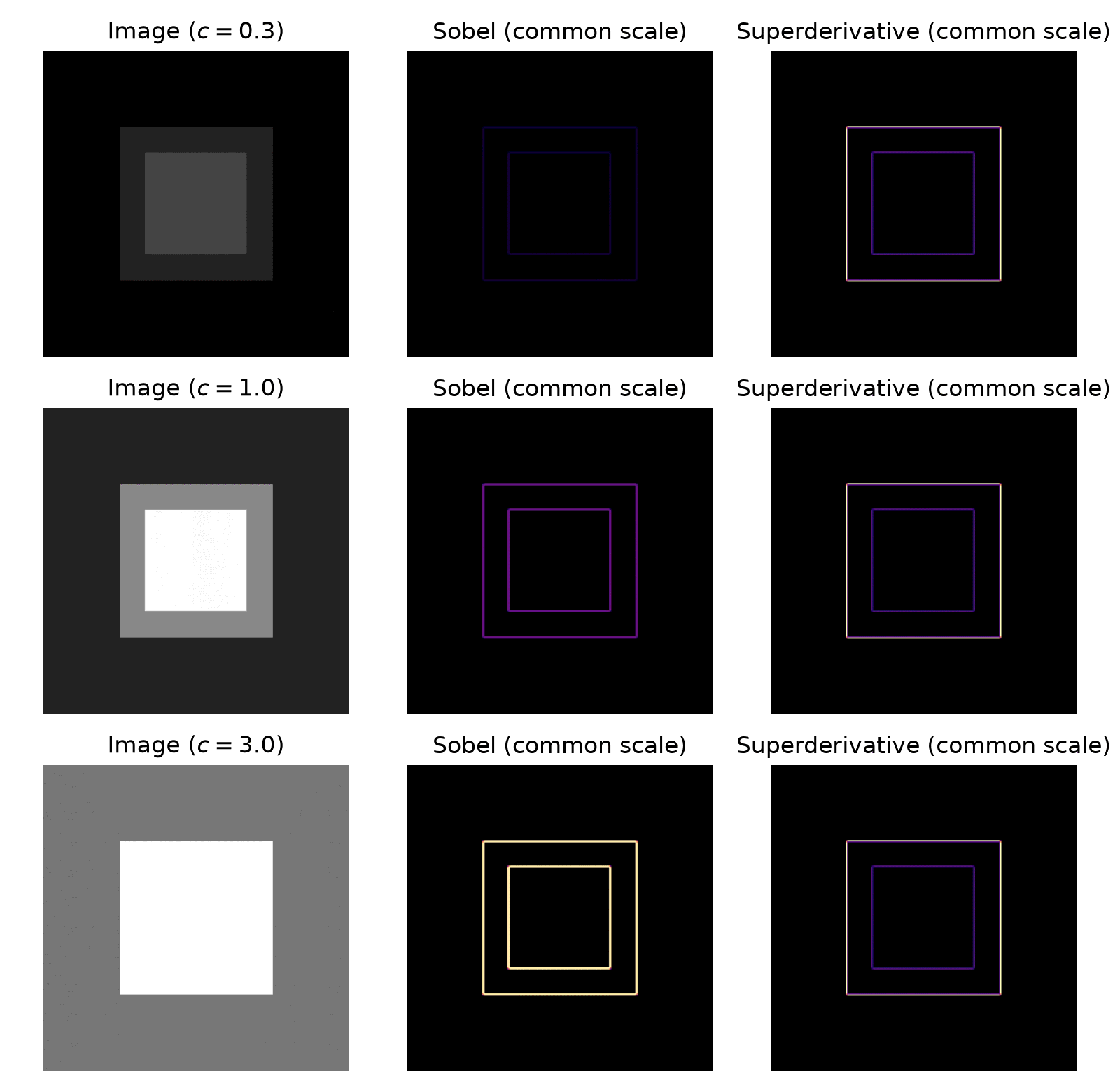}
\caption{Intensity-scale invariance.  Rows correspond to global scalings
$c=0.3,1.0,3.0$ of the same image.  Both edge maps are displayed on a common
absolute scale, without per-image normalisation.  Sobel's absolute response
grows linearly with $c$; the superderivative absolute response remains
unchanged (up to numerical round-off).  The displayed image for $c=3$ is
saturated, but the operators process the unclipped values.}
\label{fig:invariance}
\end{figure}

\subsection{Natural images}
Figure~\ref{fig:natural} shows edge-strength maps for six standard natural
images.  The superderivative operator produces responses comparable to Sobel
and Scharr, with slightly cleaner behaviour in regions of very different
brightness because of the intensity normalisation.

\begin{figure*}[!t]
\centering
\includegraphics[width=\textwidth]{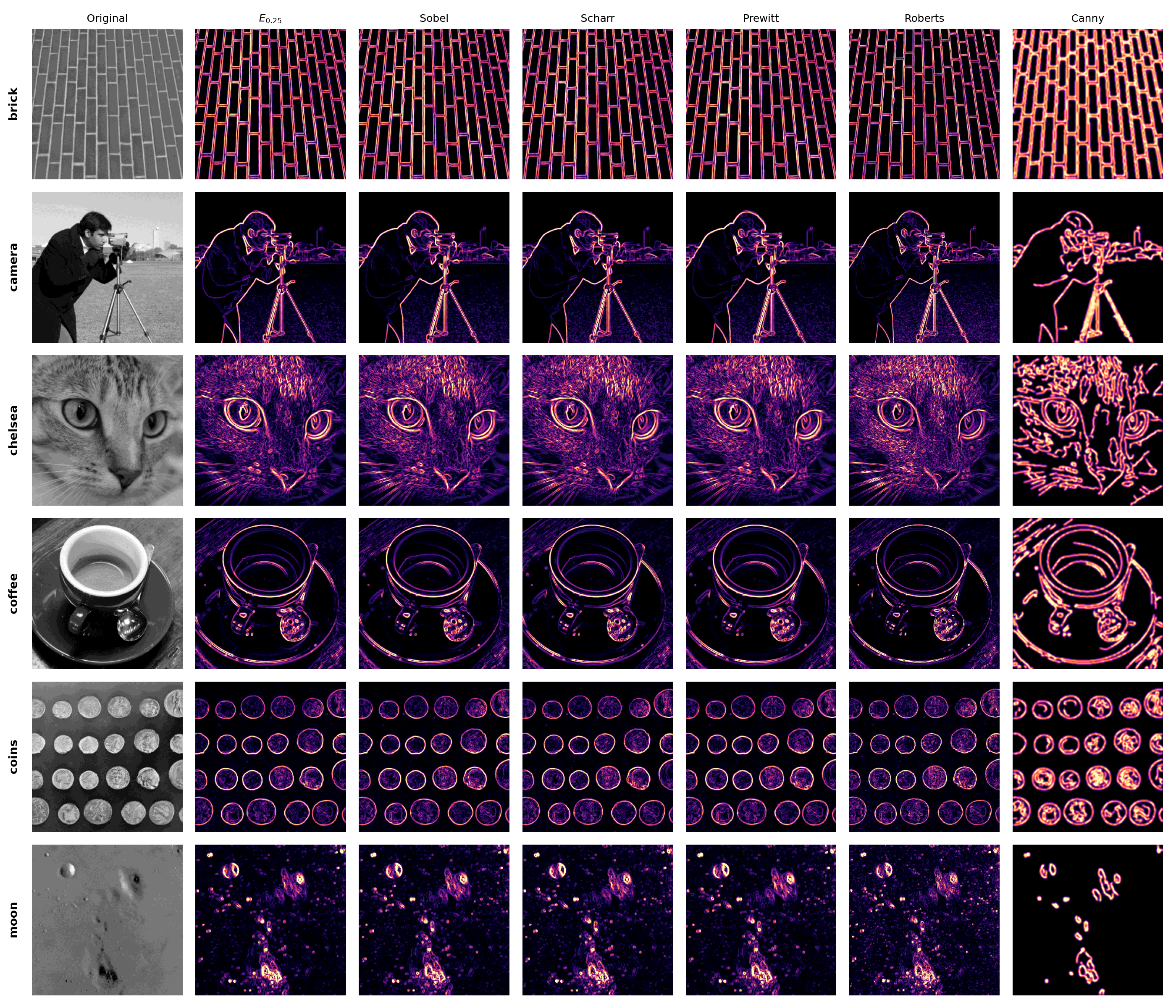}
\caption{Edge-strength maps for natural images (brick, camera, chelsea,
coffee, coins, moon).  Columns show the original image and the outputs of the
superderivative operator ($p=0.25$), Sobel, Scharr, Prewitt, Roberts and Canny.}
\label{fig:natural}
\end{figure*}

\subsection{BSDS500 benchmark}
\label{sec:bsds}
\begin{figure}[!t]
\centering
\includegraphics[width=\columnwidth]{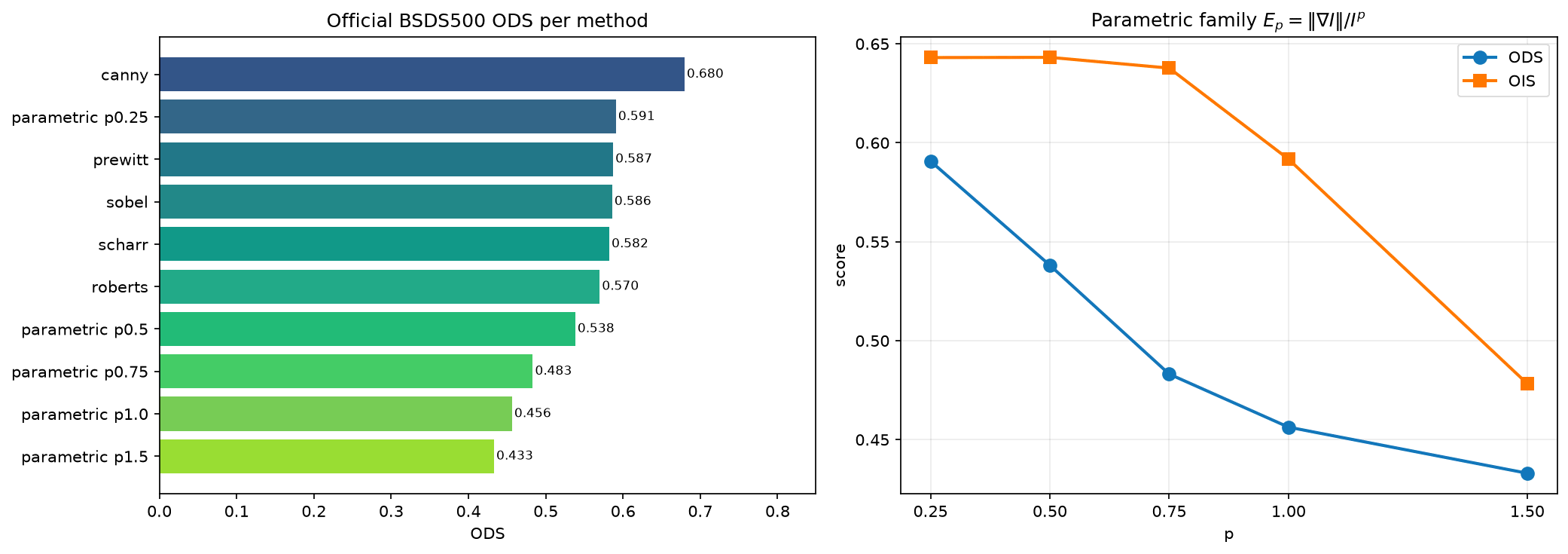}
\caption{Official BSDS500 results on the validation split.  Left: ODS for
each method.  Right: ODS and OIS as a function of the parametric exponent $p$.}
\label{fig:bsds}
\end{figure}

We evaluated the detectors on the 100 images of the BSDS500 validation split
with the official boundary benchmark of Arbel\'aez \emph{et al.}~\cite{arbelaez}.
The matching code is the original Berkeley C++ implementation (the same
\texttt{correspondPixels} used by the released benchmark), so the pixel
correspondence is the optimal, minimum-cost assignment of the Berkeley CSA
solver rather than an approximation.  Following the official defaults, each of the
five (or six) human boundary maps is used separately, the prediction is
thresholded at the 99 levels $\{0.01,\dots,0.99\}$ and thinned to a one-pixel
skeleton (scikit-image \texttt{skeletonize}), and the thresholded map is matched to
each human map within $0.0075$ of
the image diagonal with an outlier cost of $100$.  We report the optimal dataset
scale (ODS, the best F-measure obtained with a single threshold shared by all images),
the optimal image scale (OIS, the best F-measure obtained with the threshold tuned for each
image) and the average precision (AP).  ODS and OIS take the maximum over the
fixed 99-threshold grid rather than interpolating between thresholds, and AP follows
the official recall interpolation.  The
detectors are our own implementations (finite-difference kernels and the Canny of
\textsc{scikit-image}, with the automatic thresholds of that library and tuned per
image over $\sigma\in\{0.5,1,1.5,2,2.5\}$), and the input is clipped to
$[10^{-3},1]$ to keep the relative operators finite in dark regions.  Each edge
map is scaled to $[0,1]$ before thresholding, by its minimum and maximum in the
main table and by its 99th percentile in the robustness check of
Section~\ref{sec:discussion}.

Table~\ref{tab:bsds} reports the result.  The tuned Canny is the best
method in ODS and OIS ($0.679$ and $0.674$), and its ODS margin over the best
member of the parametric family is large ($+0.088$) and significant on a paired
$t$-test over the per-image OIS ($+0.031$, $p\approx10^{-6}$).  Canny's
$\sigma$ is tuned per image here, but its advantage does not come from that
tuning: with a single global $\sigma$ chosen by ODS the same detector still
reaches ODS $0.658$ and OIS $0.653$, far ahead of the unsmoothed family, so the
comparison is not an oracle.  The Canny maps of \textsc{scikit-image} are binary,
so their threshold sweep is degenerate and their ODS (pooled counts) can exceed
their OIS (mean per-image best), because the two aggregate the per-image
F-measures differently; the ordering against the continuous maps is unaffected.  Among the first-derivative kernels the parametric family peaks
at $p=0.25$ on the validation split, where it has the
best ODS ($0.591$) and AP ($0.651$) of the family and beats the raw Sobel,
Prewitt and Scharr kernels by $0.004$--$0.009$ in ODS; its per-image OIS advantage
is small but significant ($+0.005$ to $+0.006$, bootstrap 95\% CI $[+0.002,+0.009]$,
$p<0.01$), and $p=0.5$ is statistically indistinguishable
from $p=0.25$ in OIS ($-0.0001$, 95\% CI $[-0.004,+0.005]$, $p=0.96$) but has
a lower ODS ($0.538$).
Table~\ref{tab:bsds_test} repeats the comparison on the 200 held-out test
images, which were not used to choose $p$.  The per-image OIS and AP advantages
replicate ($+0.004$ to $+0.006$ in OIS, bootstrap 95\% CI $[+0.002,+0.008]$,
$p<10^{-4}$; $+0.003$ to $+0.004$ in AP), but the single-threshold ODS margin
over Sobel and Prewitt does not: on the held-out split $p=0.25$ ($0.588$) is
statistically tied with Sobel ($0.589$) and marginally below Prewitt ($0.591$),
while still above Scharr ($0.586$).  The interior optimum is therefore robust in the
per-image metric, but the headline ODS gain over the kernels is confined to the
validation split; we state both and do not claim the ODS margin as a held-out result.
The full superderivative $p=1$ is \emph{worse} than the classical kernels under
min-max, by $0.13$ in ODS and $0.046$ in OIS (bootstrap 95\% CI
$[-0.073,-0.021]$, $p\approx10^{-3}$), because dividing by
the local intensity amplifies noise in dark regions, and $p=1.5$ is worse still.
The interior optimum is robust to the normalisation.  Table~\ref{tab:bsds_robust}
repeats the benchmark with a robust percentile normalisation: the family still
peaks at a small interior exponent ($p=0.5$, ODS $0.617$, with $p=0.75$ and
$p=0.25$ at $0.613$ and $0.611$) and still beats the kernels
($0.600$--$0.601$), so the optimum is not an artefact of the min-max threshold
grid.  The collapse of the fully relative endpoint, however, is: $p=1$ is below the
kernels under min-max but statistically indistinguishable from them under percentile
(per-image OIS difference from Sobel $-0.002$, 95\% CI $[-0.015,+0.010]$,
$p=0.73$), so its ranking is a threshold-grid effect of the min-max scaling rather
than a property of the relative gradient.  The two discretisations are likewise comparable under
percentile (log-image Sobel $0.604$ against two-step Sobel $0.601$).
Two baselines that the construction invites confirm the picture.  Discretising the
log-gradient on the logarithmic image, $\|K*\log I\|$, gives ODS $0.578$ and AP
$0.626$, close to the two-step Sobel kernel, so the two-step form is not
disadvantaged by its discretisation.  The additively regularised variant
($p=0.5$, $\alpha=0.05$) recovers ODS $0.574$ with OIS $0.644$ and AP $0.649$,
above the full $p=1$ and competitive with the classical kernels, which is what the
regularised denominator is for.  Finally, the comparison with Canny is not
like-for-like: Canny adds non-maximum suppression and hysteresis to Gaussian
smoothing, whereas the family here is only a thresholded edge-strength map.  To
separate the effect of smoothing from that of the full pipeline, we apply the same
Gaussian pre-smoothing to the classical kernels and to the family, with a single
$\sigma$ chosen by ODS for each (as for the single-$\sigma$ Canny above).  This
changes the picture: $E_{0.25}$ at $\sigma=2.0$ reaches ODS $0.659$, OIS $0.684$
and AP $0.713$, against ODS $0.658$, OIS $0.653$ and AP $0.652$ for the
single-$\sigma$ Canny and ODS $0.586$, OIS $0.637$ and AP $0.647$ for unsmoothed
Sobel.  The smoothed family therefore matches the single-$\sigma$ Canny in ODS and
leads it in OIS and AP, while the tuned Canny remains ahead in ODS.  Smoothed
Sobel at $\sigma=2.0$ (ODS $0.650$, OIS $0.679$, AP $0.707$) shows that most of
the gain over the unsmoothed kernels comes from the smoothing rather than the
normalisation, the latter adding a small margin on top.  This does not isolate the
contribution of non-maximum suppression and hysteresis, which the family does not
implement.  The value of the construction is therefore the family rather than its
endpoint: a small interior intensity normalisation helps by a small but detectable
margin, the optimum is interior, and a tuned Canny remains ahead of the unsmoothed
kernels.

\begin{table}[!t]
\centering
\footnotesize
\caption{Official BSDS500 boundary benchmark on the 100 validation images
(ODS/OIS/AP; five human maps used separately; 99 thresholds; thresholded
prediction thinned to one pixel; optimal-assignment matching within $0.0075$ of
the image diagonal; outlier cost $100$; edge maps scaled to $[0,1]$ by their
minimum and maximum).  Our own detector implementations, so the values are not
directly comparable to published results for other detectors.  Higher is better;
best in \textbf{bold}.}
\label{tab:bsds}
\begin{tabular}{@{}lccc@{}}
\toprule
Method & ODS & OIS & AP \\
\midrule
Canny (tuned $\sigma$)     & \textbf{0.679} & \textbf{0.674} & 0.613 \\
Canny (single $\sigma$)     & 0.658 & 0.653 & \textbf{0.652} \\
Parametric $p=0.25$         & 0.591 & 0.643 & 0.651 \\
Prewitt                     & 0.587 & 0.638 & 0.647 \\
Sobel                       & 0.586 & 0.637 & 0.647 \\
Scharr                      & 0.582 & 0.637 & 0.645 \\
Log-image Sobel             & 0.578 & 0.633 & 0.626 \\
Regularised $p=0.5$          & 0.574 & 0.644 & 0.649 \\
Roberts                     & 0.570 & 0.629 & 0.631 \\
Parametric $p=0.5$           & 0.538 & 0.643 & 0.647 \\
Parametric $p=0.75$          & 0.483 & 0.638 & 0.641 \\
Superderivative $p=1$       & 0.456 & 0.592 & 0.644 \\
Parametric $p=1.5$           & 0.433 & 0.478 & 0.644 \\
\bottomrule
\end{tabular}
\end{table}

\begin{table}[!t]
\centering
\footnotesize
\caption{Effect of Gaussian pre-smoothing on the same BSDS500 validation split
(min-max normalisation).  A single $\sigma$ is chosen by ODS for each method, as for
the single-$\sigma$ Canny.  The unsmoothed kernels are compared against Canny, but
once both are smoothed the family matches the single-$\sigma$ Canny in ODS and leads
it in OIS and AP, while the tuned Canny remains ahead in ODS.
Higher is better; best in \textbf{bold}.}
\label{tab:bsds_smooth}
\begin{tabular}{@{}lccc@{}}
\toprule
Method & ODS & OIS & AP \\
\midrule
$E_{0.25}$, $\sigma=2.0$    & \textbf{0.659} & \textbf{0.684} & \textbf{0.713} \\
Canny (single $\sigma$)     & 0.658 & 0.653 & 0.652 \\
Sobel, $\sigma=2.0$         & 0.650 & 0.679 & 0.707 \\
$E_{0.25}$, $\sigma=1.0$    & 0.631 & 0.667 & 0.686 \\
Sobel, $\sigma=1.0$         & 0.620 & 0.662 & 0.680 \\
$E_{0.25}$, $\sigma=0.5$    & 0.609 & 0.650 & 0.662 \\
Sobel, $\sigma=0.5$          & 0.593 & 0.644 & 0.657 \\
$E_{0.25}$ (no smoothing)   & 0.591 & 0.643 & 0.651 \\
Sobel (no smoothing)        & 0.586 & 0.637 & 0.647 \\
\bottomrule
\end{tabular}
\end{table}

\begin{table}[!t]
\centering
\footnotesize
\caption{Same official BSDS500 benchmark as Table~\ref{tab:bsds} but with each
edge map scaled to $[0,1]$ by its 99th percentile instead of its minimum and
maximum.  The interior optimum of the family and its margin over the kernels
survive the change; the min-max collapse of the larger exponents does not.
Higher is better; best in \textbf{bold}.}
\label{tab:bsds_robust}
\begin{tabular}{@{}lccc@{}}
\toprule
Method & ODS & OIS & AP \\
\midrule
Canny (tuned $\sigma$)     & \textbf{0.679} & \textbf{0.673} & 0.611 \\
Canny (single $\sigma$)     & 0.653 & 0.627 & 0.619 \\
Parametric $p=0.5$           & 0.617 & 0.648 & 0.631 \\
Regularised $p=0.5$          & 0.617 & 0.648 & 0.632 \\
Parametric $p=0.75$          & 0.613 & 0.645 & 0.629 \\
Parametric $p=0.25$         & 0.611 & 0.647 & \textbf{0.638} \\
Log-image Sobel             & 0.604 & 0.637 & 0.614 \\
Parametric $p=1$             & 0.601 & 0.639 & 0.620 \\
Sobel                       & 0.601 & 0.641 & 0.633 \\
Prewitt                     & 0.601 & 0.643 & 0.633 \\
Scharr                      & 0.600 & 0.641 & 0.632 \\
Roberts                     & 0.595 & 0.634 & 0.617 \\
Parametric $p=1.5$           & 0.573 & 0.620 & 0.599 \\
\bottomrule
\end{tabular}
\end{table}

\begin{table*}[!t]
\centering
\footnotesize
\caption{Same protocol as Table~\ref{tab:bsds} on the 100 validation and the 200
held-out test images (min-max normalisation).  The exponent $p=0.25$ was chosen on
validation only.  The per-image OIS and AP advantages over the kernels replicate on
the held-out split; the ODS margin over Sobel and Prewitt does not.  Our own detector
implementations; higher is better; best in \textbf{bold}.}
\label{tab:bsds_test}
\setlength{\tabcolsep}{3pt}
\begin{tabular}{@{}lcccccc@{}}
\toprule
& \multicolumn{3}{c}{Validation (100)} & \multicolumn{3}{c}{Test (200)} \\
\cmidrule(lr){2-4}\cmidrule(lr){5-7}
Method & ODS & OIS & AP & ODS & OIS & AP \\
\midrule
Canny (tuned $\sigma$)   & \textbf{0.679} & \textbf{0.674} & 0.613 & \textbf{0.681} & \textbf{0.674} & 0.617 \\
Parametric $p=0.25$       & 0.591 & 0.643 & \textbf{0.651} & 0.588 & 0.637 & \textbf{0.646} \\
Prewitt                   & 0.587 & 0.638 & 0.647 & 0.591 & 0.633 & 0.642 \\
Sobel                     & 0.586 & 0.637 & 0.647 & 0.589 & 0.633 & 0.642 \\
Scharr                    & 0.582 & 0.637 & 0.645 & 0.586 & 0.631 & 0.641 \\
Superderivative $p=1$      & 0.456 & 0.592 & 0.644 & 0.458 & 0.599 & 0.645 \\
\bottomrule
\end{tabular}
\end{table*}

Figure~\ref{fig:bsds} summarises the results under min-max.  The parametric
family peaks at a small interior exponent under both normalisations: $p=0.25$ under
min-max and $p=0.5$ under percentile, with $p=0.75$ and $p=0.25$ close behind.
The peak is shallow in OIS, and the drop of the large exponents under min-max is
the statistically solid effect.

\subsection{Outlook: mathematical alternatives and extensions}
The framework invites several refinements without introducing post-processing;
none of them is evaluated on BSDS500 here, and they are listed as outlook rather
than results.

\textbf{Parametric exponent.} The family $E_p=\|\nabla I\|/I^p$ lets the
user choose the degree of intensity normalisation.  The BSDS500 experiments
suggest $p\in[0.25,0.5]$ as a good compromise between invariance and
stability.

\textbf{Regularised denominator.} Replacing $I$ by $I+\varepsilon$ avoids
overflow in dark regions.  The resulting map
$E_{\varepsilon}=\|\nabla I\|/(I+\varepsilon)$ tends to the log-gradient when
$\varepsilon\ll I$ and to a scaled absolute gradient, $\|\nabla I\|/\varepsilon$,
when $\varepsilon\gg I$; it is a one-parameter family distinct from $E_p$, and the
two can be combined as $\|\nabla I\|/(I^p+\alpha)$.  The additively regularised
variant is included in the BSDS500 comparison above.

\textbf{Radial score for symmetric objects.} For images with approximate
radial symmetry, the score $L=r\,\partial_r\log I$ or its normalised version
$|L|/r$ can be used directly.  It measures relative change under a zoom and is
therefore naturally scale-covariant after recentring: under a spatial
rescaling it is reparametrised, not multiplied by a factor, so its value at
corresponding points is unchanged.

\textbf{Higher-order operators.} Just as the Laplacian extends the gradient,
one can take the divergence of the relative-gradient field or apply the
multiplicative derivative twice, giving a relative-curvature operator.  A
systematic treatment is left for future work; the expression
$\partial_x(x\partial_x\log I)+\partial_y(y\partial_y\log I)$ used here is only
a candidate and is not the second bigeometric derivative.

\textbf{Scale-space.} All kernels in Table~\ref{tab:kernels} can be applied
after Gaussian smoothing at scale $\sigma$.  A multi-scale superderivative
operator would compute $E_p$ at several scales and combine the responses.

\subsection{Parametric family visualisation}
Figure~\ref{fig:parametric} shows the parametric edge-strength maps
$E_p=\|\nabla I\|/I^p$ for several values of $p$.  As $p$ increases, the
response becomes more uniform across regions of different brightness, but the
division by smaller intensities also amplifies noise in dark areas.

\begin{figure*}[!t]
\centering
\includegraphics[width=\textwidth]{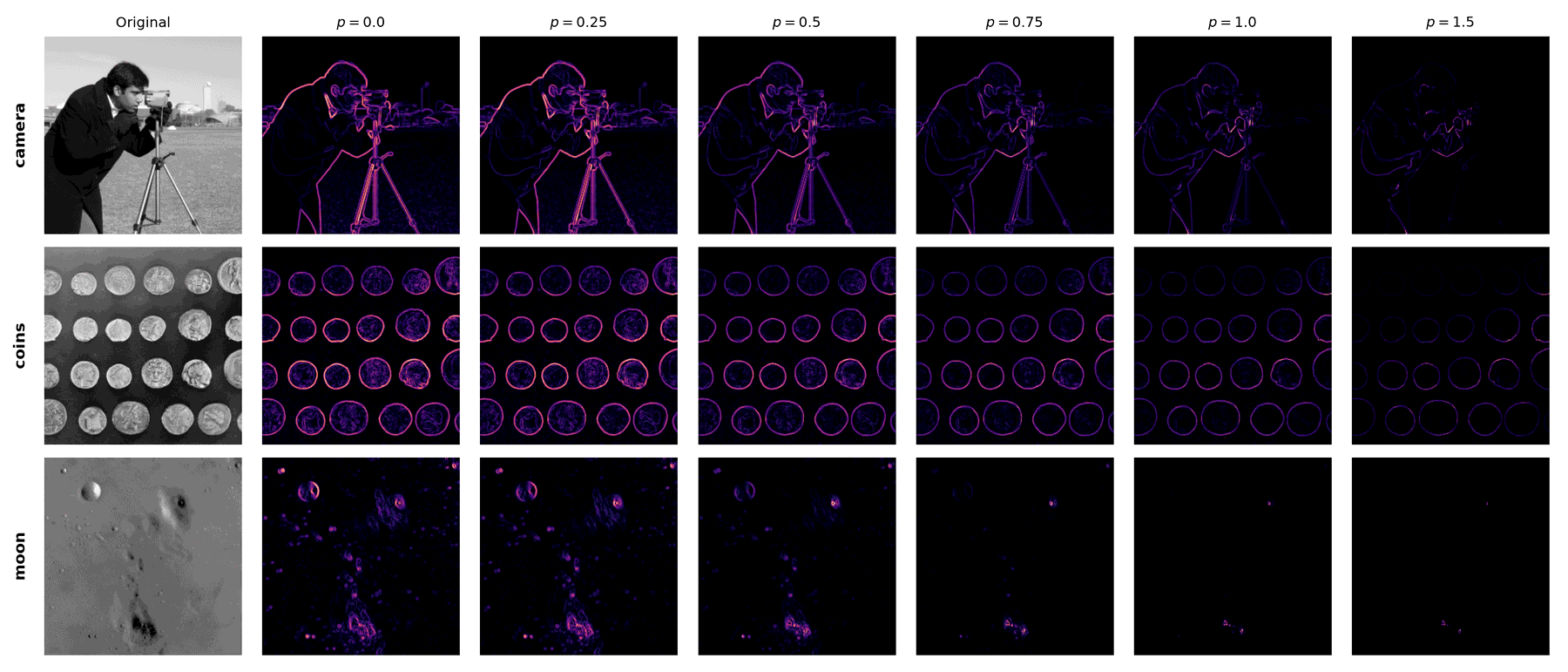}
\caption{Parametric family $E_p=\|\nabla I\|/I^p$ for $p=0,0.25,0.5,0.75,1,1.5$
on natural images.  $p=0$ is the absolute gradient; $p=1$ is the
superderivative/log-gradient.}
\label{fig:parametric}
\end{figure*}

\subsection{Noise robustness}
Figure~\ref{fig:noise} compares the detectors as additive Gaussian noise
increases.  Because the relative operator divides by the local intensity, additive
noise is amplified where $I$ is small, so its advantage over the classical
kernels at low noise reverses at high noise: measured as best F1 against the
synthetic ground truth of Fig.~\ref{fig:edges}, at $\sigma=0.02$ the relative
$p=0.5$ map is marginally ahead of Sobel ($1.000$ vs.\ $0.998$), at
$\sigma=0.05$ they are comparable ($1.000$ vs.\ $0.998$), and at
$\sigma=0.10$ the relative map collapses ($0.265$ for $p=0.5$, $0.038$ for
$p=1$) while Sobel holds at $0.992$.  This is the same dark-region
amplification that limits $p=1$ on BSDS500, and it is what the regularised
denominator of Section~\ref{sec:regularised} is for; the low-light experiment
below makes the point in the spatial domain.

\begin{figure*}[!t]
\centering
\includegraphics[width=\textwidth]{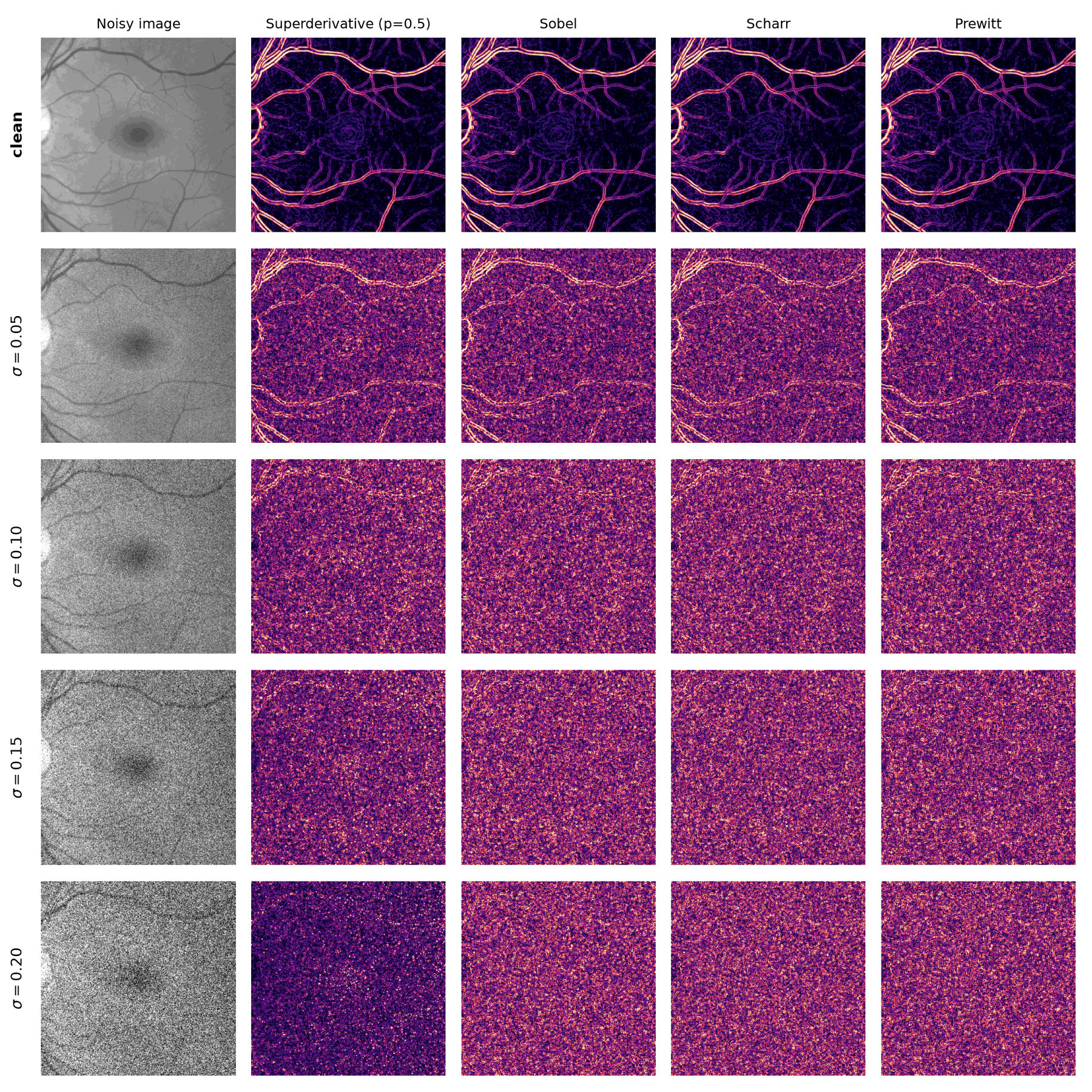}
\caption{Robustness to additive Gaussian noise.  Rows correspond to increasing
noise standard deviation; columns show the noisy input and the outputs of the
superderivative ($p=0.5$), Sobel, Scharr and Prewitt operators.}
\label{fig:noise}
\end{figure*}

\subsection{Low-light and texture}
Figures \ref{fig:lowlight} and \ref{fig:texture} show two additional regimes:
images under simulated low-light conditions and highly textured natural
surfaces.  Under simulated low light the proposed operator degrades less than the
classical kernel in the map-difference sense: the mean absolute difference between
the normal-light and low-light maps, relative to the normal-light mean, is $0.249$ at
$p=0.5$ and $0.235$ for the regularised variant, against $0.321$ for Sobel,
while the fully relative $p=1$ is the worst at $0.603$.  This measure is a
self-consistency check, not an accuracy metric, and low light here is a gamma
compression plus additive noise rather than the global rescaling of the invariance
proposition.  On textured images the operator likewise suppresses weak texture
gradients through the same normalisation, at the cost of some faint true edges.

\begin{figure*}[!t]
\centering
\includegraphics[width=\textwidth]{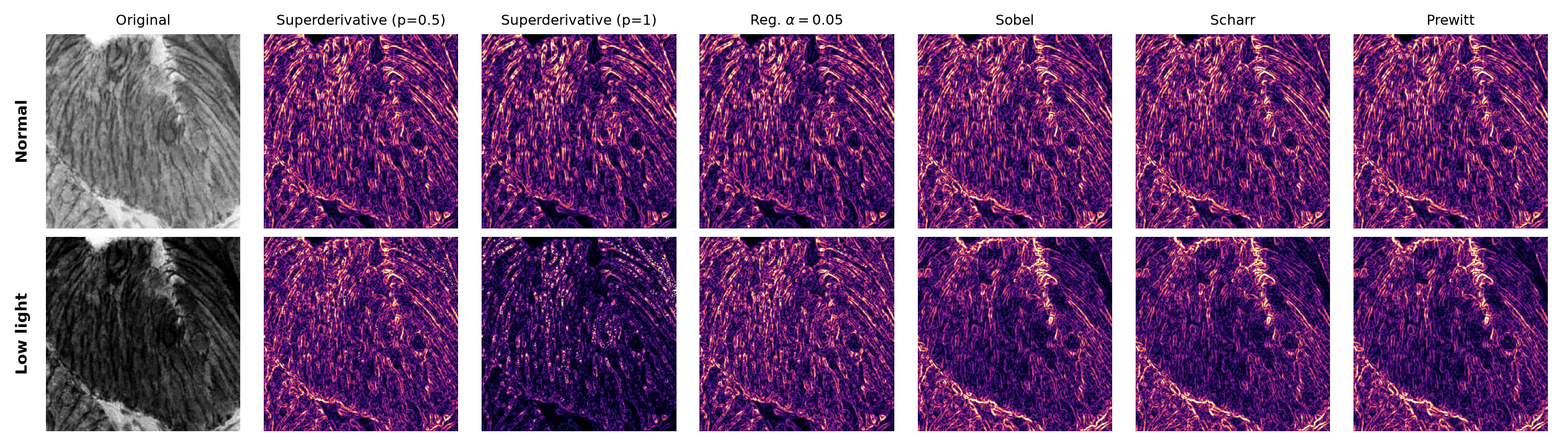}
\caption{Edge detection under normal (top row) and simulated low-light
(bottom row) conditions.  Low light is a gamma compression $I^{2.5}$ followed by
additive sensor noise $\sigma=0.02$, \emph{not} the global rescaling
$I\mapsto cI$ of the invariance proposition, under which the maps are identical
to machine precision (Fig.~\ref{fig:invariance}).  What degrades the response is
the noise, not the loss of light.  Figures are means over ten noise realisations
with a fixed seed (42); the spread is $\pm0.002$ for $p=0.5$ and the regularised
variant and $\pm0.021$ for $p=1$.}
\label{fig:lowlight}
\end{figure*}

\begin{figure*}[!t]
\centering
\includegraphics[width=\textwidth]{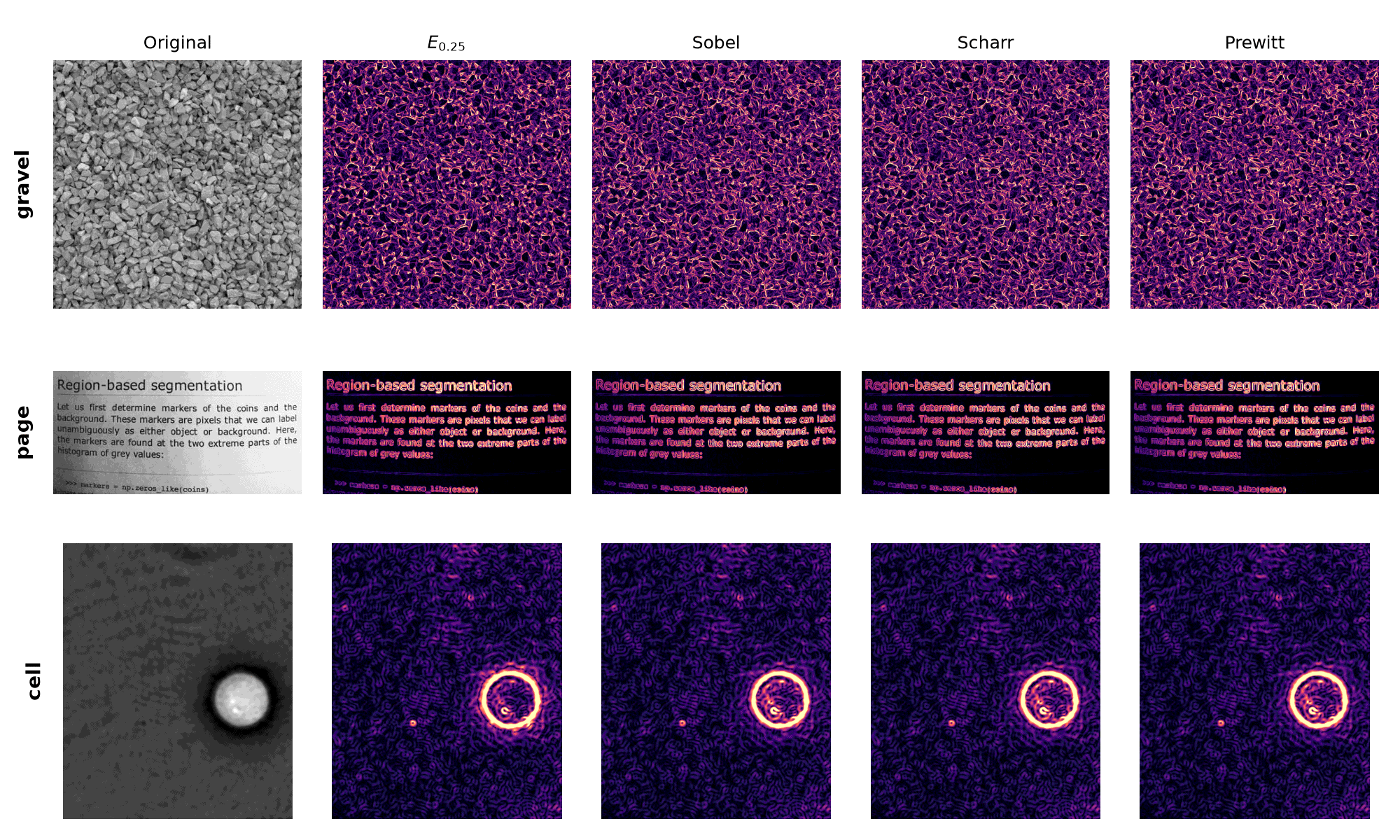}
\caption{Edge detection on textured images (gravel, page, cell).  The
superderivative column uses $p=0.25$.}
\label{fig:texture}
\end{figure*}

\subsection{Controlled checks of relative versus absolute change}
\label{sec:relative}
Because the global rescaling of the invariance proposition cancels under any
per-image normalisation, the invariance is only observable without normalisation
(Fig.~\ref{fig:invariance}).  Its practical content is that equal \emph{contrast
ratios} $\Delta I/I$ receive equal response.  Figure~\ref{fig:relative} and the
numbers below isolate that content on two analytic images.  On two steps of the same
absolute size $0.05$ but on backgrounds $0.10$ and $0.50$, the ratio of the
bright-step to dark-step peak response is $1.000$ for Sobel and $0.273$ for
$E_{\mathrm{SD}}$; on two steps of the same contrast ratio $0.20$ but on the same
two backgrounds, it is $5.000$ for Sobel and $1.000$ for $E_{\mathrm{SD}}$.  On a
reflectance image with faint edges, $I=R\cdot L$, multiplied by a smooth
illumination field $L$, the best F1 against the edges of $R$ is $1.000$ for
$E_{\mathrm{SD}}$ ($p=1$) against $0.998$ for Sobel, and $E_{\mathrm{SD}}$ is
already at the level reached by Sobel after exact division by $L$.  The gain is small
for a smooth field, as expected, and grows with the dynamic range of the
illumination; it is this contrast equalisation, not the global rescaling of
Fig.~\ref{fig:invariance}, that the relative normalisation buys.

\begin{figure*}[!t]
\centering
\includegraphics[width=\textwidth]{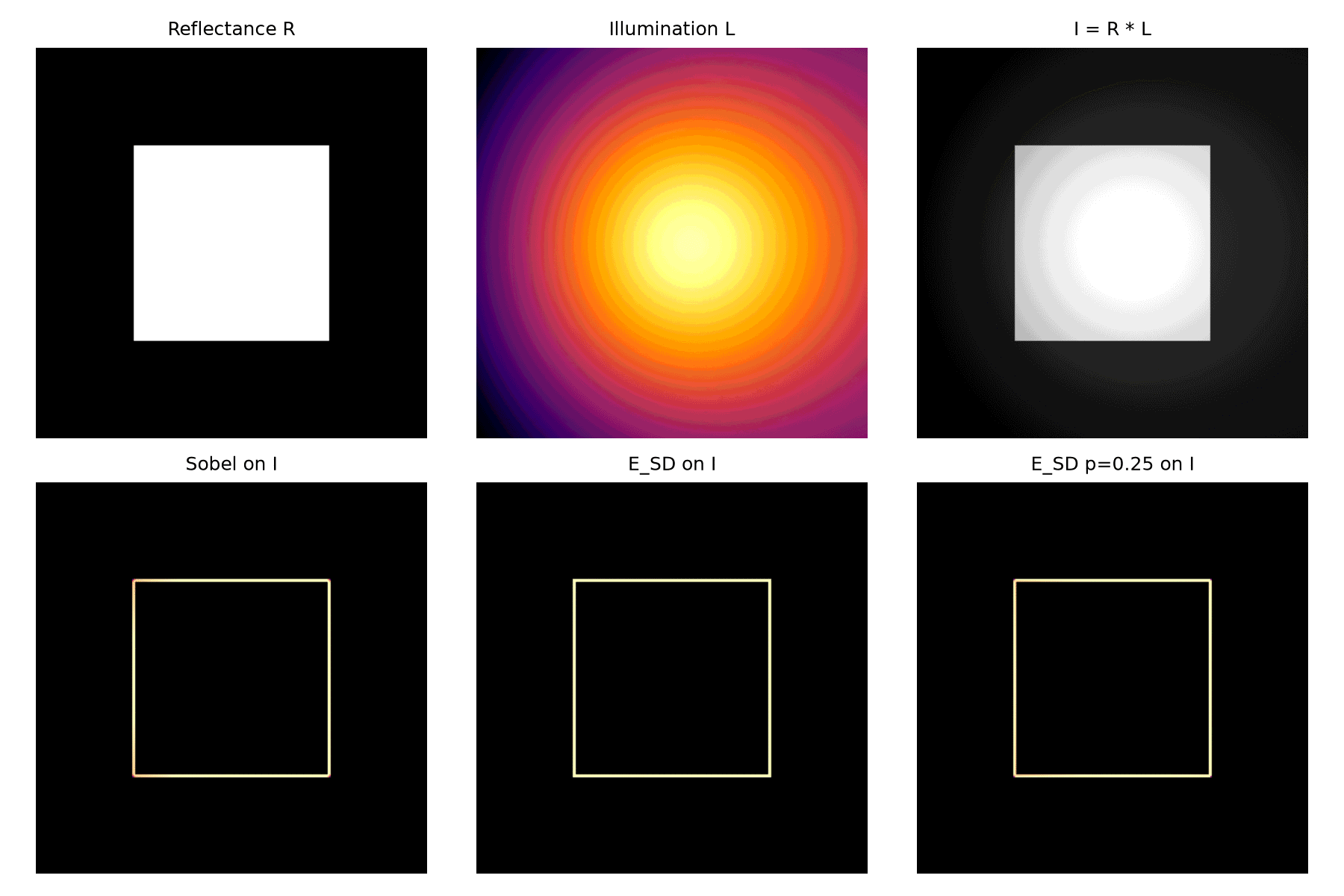}
\caption{Controlled checks of the relative change on analytic images.  Top row:
reflectance $R$, smooth illumination field $L$ and the image $I=R\cdot L$; bottom
row: Sobel, $E_{\mathrm{SD}}$ ($p=1$) and $E_{0.25}$ on $I$, all scaled to
their 99th percentile.  Sobel also responds to the illumination gradient,
whereas the relative operators are close to the reflectance edges.}
\label{fig:relative}
\end{figure*}

\subsection{Discretisation of the relative gradient}
\label{sec:discretisation}
The two discretisations of $\|\nabla\log I\|$ agree away from a step but not at
one.  On the synthetic piecewise-constant images of Fig.~\ref{fig:edges}, the
two-step form $\|K*I\|/I$ reaches best F1 $1.000$ and the log-image form
$\|K*\log I\|$ reaches $0.998$: the two are within the tolerance of the
benchmark, so the two-step form used throughout is a sound discretisation and the
small difference is not what drives the $p=1$ behaviour on BSDS500, which is a
dark-region effect rather than a discretisation artefact.

\section{Discussion}
\label{sec:discussion}
On the official BSDS500 validation split a small interior exponent
($p\approx0.25$--$0.5$) gives a small but statistically detectable improvement
over the classical first-derivative kernels under both a min-max and a robust
percentile normalisation of the edge maps, so the interior optimum is not an artefact
of the threshold grid.  On the held-out test split the per-image OIS and AP advantages
replicate, while the single-threshold ODS margin over Sobel and Prewitt does not
(Table~\ref{tab:bsds_test}), so the positive claim is stated per-image and the ODS
margin is confined to the validation split.  The fully relative operator ($p=1$) is never the best member
and its ranking is normalisation-sensitive: division by a small $I$ amplifies
dark-region noise, and the min-max scaling of a heavy-tailed map turns that noise
into a threshold-grid effect that the percentile normalisation removes.  On the
synthetic piecewise-constant geometry most first-derivative methods recover the
edges almost equally well, with Roberts on the triangle as the exception.  The
conceptual advantage of the
framework---making the absolute-to-relative passage explicit---and the verified
intensity-scale invariance at $p=1$ (Fig.~\ref{fig:invariance}) remain, but they do
not by themselves select the empirically best member of the family.

Its main limitation is the division by $I$.  In very dark regions a small
amount of noise is amplified because the denominator is close to zero.  In
practice this is mitigated by adding a small offset $\varepsilon$ to the
denominator, at the cost of a slight loss of exact invariance in the darkest
areas.  A second limitation is that the full radial score $L$ is not
shift-invariant, because the weights $x$ and $y$ are measured from the image
centre.  The isotropic version $E_{\mathrm{SD}}$ restores shift invariance
and is the recommended form for general edge detection.  A third, methodological
limitation is that neither the unsmoothed nor the merely smoothed first-derivative
kernels are like-for-like comparisons with Canny, which also applies non-maximum
suppression and hysteresis.  Table~\ref{tab:bsds_smooth} shows that once the same
Gaussian smoothing is applied, the family matches the single-$\sigma$ Canny in ODS and
leads it in OIS and AP, but the tuned Canny remains ahead in ODS; the comparison
therefore separates the effect of smoothing from that of the full Canny pipeline only
partially.

Future work includes adding non-maximum suppression and hysteresis
thresholding (as in Canny), mathematical robustifications of the general
formulation and its discretisation, multi-scale extensions, and applying the
operator to image enhancement, segmentation and feature detection.

\section{Conclusions}
\label{sec:conclusions}
We have revisited the log-gradient edge operator and motivated it by analogy with
the superderivative, a multiplicative derivative in which the independent
variable is perturbed multiplicatively.  The operator measures relative rather
than absolute intensity change, is invariant under global intensity rescaling,
and can be implemented in two simple steps.  A parametric family interpolates
between the absolute gradient and the log-gradient; on synthetic images most
first-derivative methods recover the edges almost perfectly, with Roberts on the
triangle as the exception, while on the official
BSDS500 benchmark the best members are small intermediate exponents
($p\approx0.25$--$0.5$), edging the classical kernels by a small but detectable
margin in the per-image metric; the single-threshold ODS margin over the kernels holds
on the validation split but not on the held-out test split, and the fully relative
endpoint $p=1$ is never the best member, its ranking depending on the edge-map
normalisation.  With the same Gaussian smoothing that
the single-$\sigma$ Canny uses, the smoothed family matches it in ODS and leads it
in OIS and AP, while the tuned Canny remains ahead in ODS; the unsmoothed
comparison therefore understates the smoothed family, though it does not isolate the
effect of non-maximum suppression and hysteresis.  The contribution is therefore
the construction, the invariance it buys, and a clear account of when it helps --- a
small gain at $p\approx0.25$ and a loss at $p=1$ --- not
a new state of the art in boundary detection.  The natural next step is to embed
$E_p$ in a full edge pipeline with smoothing, non-maximum suppression and
hysteresis thresholding and to re-evaluate whether the same interior exponent is
preferred.

\section*{Data and code availability}
The BSDS500 dataset is publicly available from the Berkeley Computer Vision
Group.  The reference implementation, the synthetic evaluation, and the official
BSDS500 boundary evaluation (\texttt{src/official\_bsds/}), which links the
unmodified Berkeley matching code, together with the frozen results reported in this
paper, are available at
\url{https://github.com/agarnung/computervisionmiscellaneous} (directory
\texttt{papers/superderivative}).  The frozen result files
\path{src/synthetic_results.json},
\path{src/bsds_results_official_code.json} (validation split),
\path{src/bsds_results_official_code_test.json} (held-out test split) and
\path{src/bsds_results_official_code_percentile.json} (percentile
normalisation), which the evaluation and the synthetic benchmark regenerate
deterministically.


\end{document}

%% file: paper_style.tex
\usepackage{float}
\usepackage{placeins}
\usepackage{caption}

%% file: main.bbl
\begin{thebibliography}{10}
\bibitem{fechner}
G.~T.~Fechner, \emph{Elemente der Psychophysik}. Leipzig, Germany:
Breitkopf und H\"artel, 1860.  (Fechner's logarithmic law; the Weber
contrast $\Delta I/I$ is named after E.~H.~Weber.)

\bibitem{reinhard}
E.~Reinhard \emph{et al.}, ``Photographic tone reproduction for digital images,''
\emph{ACM Trans. Graph.}, vol.~21, no.~3, pp.~267--276, 2002.

\bibitem{land1971}
E.~H.~Land and J.~J.~McCann, ``Lightness and retinex theory,''
\emph{J. Opt. Soc. Amer.}, vol.~61, no.~1, pp.~1--11, 1971.

\bibitem{land1977}
E.~H.~Land, ``The retinex theory of color vision,''
\emph{Sci. Amer.}, vol.~237, no.~6, pp.~108--128, 1977.

\bibitem{jobson1997}
D.~J.~Jobson, Z.~Rahman, and G.~A. Woodell, ``A multiscale retinex for
bridging the gap between color images and the human observation of scenes,''
\emph{IEEE Trans. Image Process.}, vol.~6, no.~7, pp.~965--976, 1997.

\bibitem{kimmel2003}
R.~Kimmel, M.~Elad, D.~Shaked, R.~Keshet, and I.~Sobel, ``A variational
framework for retinex,'' \emph{Int. J. Comput. Vis.}, vol.~52, no.~1,
pp.~7--23, 2003.

\bibitem{homomorphic}
A.~V. Oppenheim, R.~W. Schafer, and T.~G. Stockham, ``Nonlinear filtering of
multiplied and convolved signals,'' \emph{Proc. IEEE}, vol.~56, no.~8,
pp.~1264--1291, 1968.

\bibitem{grossman}
M.~Grossman and R.~Katz, \emph{Non-Newtonian Calculus}. Pigeon Cove, MA:
Lee Press, 1972.

\bibitem{zotkin}
A.~Zotkin, ``Supercomputaci\'on tetr\'adica: primera aproximaci\'on hacia una
teor\'ia de la super-relatividad,'' \emph{TARD\'IGRADOS} (blog), Apr.~2018.
[Online]. Available:
\url{https://tardigrados.wordpress.com/2018/04/22/supercomputacion-tetradica-primera-aproximacion-hacia-una-teoria-de-la-super-relativididad/}.
This is the source of the term \emph{superderivative}; the note is not
peer-reviewed, and the same operator appears as the bigeometric derivative in
multiplicative calculus~\cite{grossman}.

\bibitem{jourlin}
M.~Jourlin, \emph{Logarithmic Image Processing: Theory and
Applications}, Advances in Imaging and Electron Physics, vol.~195. Amsterdam,
The Netherlands: Elsevier, 2016.

\bibitem{florack}
L.~Florack and H.~van~Assen, ``Multiplicative calculus in biomedical image
analysis,'' \emph{J. Math. Imaging Vis.}, vol.~42, no.~1, pp.~64--75, 2012.

\bibitem{kovesi}
P.~Kovesi, ``Phase congruency: A low-level image invariant,''
\emph{Psychological Research}, vol.~64, no.~2, pp.~136--148, 2000.

\bibitem{touzi}
R.~Touzi, A.~Lopes, and P.~Bousquet, ``A statistical and geometrical edge
detector for SAR images,'' \emph{IEEE Trans. Geosci. Remote Sens.}, vol.~26,
no.~6, pp.~764--773, 1988.

\bibitem{gonzales}
R.~C. Gonzalez and R.~E. Woods, \emph{Digital Image Processing}, 4th~ed. Pearson, 2018.
\bibitem{roberts}
L.~G.~Roberts, ``Machine perception of three-dimensional solids,''
Ph.D. dissertation, Massachusetts Institute of Technology, 1963.

\bibitem{prewitt}
J.~M.~S.~Prewitt, ``Object enhancement and extraction,'' in \emph{Picture
Processing and Psychopictorics}, B.~Lipkin and A.~Rosenfeld, Eds.  New York:
Academic Press, 1970, pp.~75--149.

\bibitem{sobel}
I.~Sobel and G.~Feldman, ``A $3\times3$ isotropic gradient operator for image
processing,'' presented at the Stanford Artificial Intelligence Project, 1968.

\bibitem{freichen}
W.~Frei and C.-C.~Chen, ``Fast boundary detection: A generalization and a new
algorithm,'' \emph{IEEE Trans. Comput.}, vol.~C-26, no.~10, pp.~988--998, 1977.

\bibitem{scharr}
H.~Scharr, ``Optimal operators in digital image processing,''
Ph.D. dissertation, Universit{\"a}t Heidelberg, 2000.

\bibitem{kirsch}
R.~Kirsch, ``Computer determination of the constituent structure of
biological images,'' \emph{Comput. Biomed. Res.}, vol.~4, no.~3, pp.~315--328,
1971.

\bibitem{robinson}
G.~S.~Robinson, ``Color edge detection,''
\emph{Proc. SPIE}, vol.~87, pp.~126--133, 1976.

\bibitem{farid}
H.~Farid and E.~P.~Simoncelli, ``Differentiation of discrete
multidimensional signals,''
\emph{IEEE Trans. Image Process.}, vol.~13, no.~4, pp.~496--508, 2004.

\bibitem{canny}
J.~Canny, ``A computational approach to edge detection,''
\emph{IEEE Trans. Pattern Anal. Mach. Intell.}, vol.~PAMI-8, no.~6, pp.~679--698, 1986.

\bibitem{farid1997}
H.~Farid and E.~P.~Simoncelli, ``Optimally rotation-equivariant directional
derivative kernels,'' in \emph{Computer Analysis of Images and Patterns (CAIP)},
LNCS, 1997, pp.~207--214.

\bibitem{drago2003}
F.~Drago, K.~Myszkowski, T.~Annen, and N.~Chiba, ``Adaptive logarithmic
mapping for displaying high contrast scenes,'' \emph{Computer Graphics Forum},
vol.~22, no.~3, pp.~419--426, 2003.

\bibitem{bashirov2008}
A.~E.~Bashirov, E.~M.~Kurp{\i}nar, and A.~\"Ozyap{\i}c{\i}, ``Multiplicative
calculus and its applications,'' \emph{J. Math. Anal. Appl.}, vol.~337, no.~1,
pp.~36--48, 2008.

\bibitem{lindeberg1994}
T.~Lindeberg, ``Scale-space theory: A basic tool for analyzing structures at
different scales,'' \emph{J. Appl. Stat.}, vol.~21, no.~1--2, pp.~225--270, 1994.

\bibitem{arbelaez}
P.~Arbel\'aez, M.~Maire, C.~Fowlkes, and J.~Malik, ``Contour detection and
hierarchical image segmentation,'' \emph{IEEE Trans. Pattern Anal. Mach.
Intell.}, vol.~33, no.~5, pp.~898--916, 2011.

\bibitem{samko1993}
S.~G. Samko, A.~A. Kilbas, and O.~I. Marichev,
\emph{Fractional Integrals and Derivatives: Theory and Applications}.
Yverdon, Switzerland: Gordon and Breach, 1993.

\bibitem{kilbas2006}
A.~A. Kilbas, H.~M. Srivastava, and J.~J. Trujillo,
\emph{Theory and Applications of Fractional Differential Equations}.
Amsterdam, The Netherlands: Elsevier, 2006.

\bibitem{miller1993}
K.~S. Miller and B.~Ross,
\emph{An Introduction to the Fractional Calculus and Fractional Differential
Equations}.
New York, NY: Wiley, 1993.

\end{thebibliography}
